\documentclass[journal,onecolumn]{IEEEtran}
\IEEEoverridecommandlockouts

\usepackage[colorlinks=false]{hyperref}
\usepackage[dvips]{graphicx}
\usepackage[usenames,dvipsnames]{color}
\usepackage{epsfig}
\usepackage{rotating}
\usepackage{amssymb}
\usepackage{amsmath,amsfonts}
\usepackage{lscape}
\usepackage[latin1]{inputenc}
\usepackage{verbatim}
\usepackage[tableposition=top,font=small,labelfont=bf,format=hang]{caption}
\usepackage{booktabs}
\newtheorem{lem}{Lemma}
\newtheorem{thm}{Theorem}
\newtheorem{prop}{Proposition}

\newtheorem{cor}{Corollary}
\newtheorem{rem}{Remark}
\newtheorem{exam}{Example}
\newcommand{\F}{\mathbb{F}}

\newenvironment{proof}{{\em Proof:}}
{\hspace{\stretch{1}}%
\rule{1ex}{1ex}\\}

\begin{document}

\title{Construction of optimal Hermitian self-dual codes from unitary matrices}
\author
{
{Lin Sok, School of Mathematical Sciences, Anhui University, Hefei, Anhui, 230601, China and}\\{ Department of Mathematics, Royal University of Phnom Penh, 12156 Phnom Penh, Cambodia, \\{ \tt soklin\_heng@yahoo.com}}
}
\date{}
\maketitle

We provide an algorithm to construct unitary matrices over finite fields. We present various constructions of Hermitian self-dual code by means of unitary matrices, where some of them generalize the quadratic double circulant constructions. Many optimal Hermitian self-dual codes over large finite fields with new parameters  are obtained. More precisely MDS or almost MDS Hermitian self-dual codes of lengths up to $18$ are constructed over finite fields $\F_{q},$ where $q=3^2,4^2,5^2,7^2,8^2,9^2,11^2,13^2,17^2,19^2.$ Comparisons with classical constructions are made.

\section{Introduction}\label{section:intro}
MDS codes form an optimal family of classical codes. They are closely related to combinatorial designs and finite geometry and have many applications in both theory and practice.

Self-dual codes are one of the most interesting classes of linear codes that find various applications in cryptographic protocols (secret sharing schemes) and combinatorics. They have close connections with group theory, lattice theory and design theory. It is well known that Euclidean binary self-dual codes are asymptotically good~\cite{MacSloTho}. Constructions of Euclidean self-dual codes over large finite fields were given by many authors \cite{AraGul,BetGeoMas,Gabo,GabOtm,GeoKou,GraGul,Gue,JinXin,KimLee,KimLee07,ShiSokSole}.

There have been a lot of works on Euclidean self-dual codes but less results have been known for the Hermitian case. The motivation of studying Hermitian self-dual (self-orthogonal) codes over $\F_{q^2}$ was due to their connection to $q-$ary quantum stabilizer codes \cite{AshKni}. Quaternary Hermitian self-dual codes were considered by MacWilliams et al.\cite{MacOdlSloWar} where they first gave the classification of length up to $14.$ Following the work \cite{MacOdlSloWar}, Conway et al. \cite{ConPleSlo} completed the classification of quaternary Hermitian self-dual codes of length up to $16.$ Later Huffman \cite{Huf90,Huf91,Huf97} classified the extremal quaternary Hermitian self-dual codes of length up to $28$. In 2011, Harada et al. \cite{HarMun} classified all quaternary Hermitian self-dual codes of length $20.$

There are two well-known constructions of Hermitian self-dual codes over large finite fields; the quadratic double circulant construction by Gaborit \cite{Gabo} and the building-up construction by Kim et al. \cite{KimLee}. 

Recently Tong and Wang \cite{TongWang} have constructed all MDS $q^2-$ary Hermitian self-dual codes of all lengths less than or equal to $q+1$ from the generalized Reed-Solomon codes. 

In this paper we first study unitary groups over finite fields. We provide methods to construct unitary matrices and apply them to construct Hermitian self-dual codes over finite fields of prime power orders. We obtain many new optimal codes, more precisely MDS or almost MDS Hermitian self-dual codes of lengths up to 18 are constructed over finite fields with sizes $q=3^2,4^2,5^2,7^2,9^2,11^2,13^2,16^2,17^2,19^2$. Further more with the same lengths, our method can also be applied efficiently for any $q$ being a square greater than $19^2$. Some MDS and almost MDS as well as optimal Hermitian self-dual codes with new parameters are summarized in Table \ref{table:new00} and  Table \ref{table:new01}.

For rather small values $q=2^2,3^2,4^2,5^2$, we construct optimal Hermitian self-dual codes up to length $28$ (up to length $34$ for $q=2^2$). Numerical results show that our constructions perform better than the quadratic double circulant contruction \cite{Gabo} and the building-up construction \cite{KimLee}, for example a Hermitian self-dual code over $
\F_4$ obtained from our construction has parameters $[18,9,8]$ which are better than those of \cite{Gabo} and Hermitian self-dual codes with parameters $[10,5,6],[10,5,6],[12,6,7]$ over $\F_{3^2},\F_{4^2},\F_{ 8^2}$ respectively are better than those of \cite{KimLee} with the parameters over $\F_{8^2}$ being new in \cite{TongWang}. Furthermore, over $\F_{q},q=11^2,13^2,16^2,17^2$, we obtain MDS Hermitian self-dual codes with parameters $[14,7,8]$ which are better than \cite{GulKimLee} while the parameters over $\F_{11^2}$ are new in \cite{TongWang}. On the one hand, our construction has no restriction on lengths like in \cite{Gabo} and thus more parameters are available. On the other hand, our method is easily applicable to construct codes of large lengths without going any recursive step like in \cite{KimLee}, which makes the code construction faster. All the computations are done with Magma \cite{mag}.

The paper is organized as follows: Section \ref{section:pre} gives preliminaries and background on self-dual codes as well as studies the unitary group over finite fields. Section \ref{section:con} gives different constructions of Hermitian self-dual codes.  Section \ref{section:matrix-product} studies matrix product codes which are Hermitian self-dual. Section \ref{section:embedding} provides a method to embed a self-orthogonal code into a self-dual code.
Section \ref{section:num} describes parameters of different constructions and makes comparisons among them. We end up with some concluding remarks in Section \ref{section:conclusion}.

\section{Preliminaries}\label{section:pre}
\label{Sec-Prelim}

\subsection{Self-dual codes}
A {\em
linear $[n,k]$ code $C$ of length $n$} over ${{\mathbb F}_q}$ is a $k$-dimensional subspace
of  $ {\mathbb F}_q^n$. An element in $C$ is called a {\em codeword}. The
(Hamming) weight wt$({\bf{x}})$ of a vector ${\bf{x}}=(x_1, \dots,
x_n)$ is the number of non-zero coordinates in it. The {\em minimum
distance ({\rm{or}} minimum weight) $d(C)$} of $C$ is
$d(C):=\min\{{\mbox{wt}}({\bf{x}})~|~ {\bf{x}} \in C, {\bf{x}} \ne
{\bf{0}} \}$. For ${\bf{x}}=(x_1,
\dots, x_n)$ and ${\bf{y}}=(y_1, \dots, y_n)$ in ${\mathbb F}_{q^2}^n$, their {\em Euclidean and Hermitian inner product} are defined respectively by
$${\bf{x}}\cdot{\bf{y}}=\sum_{i=1}^n x_i y_i, {\bf{x}}*{\bf{y}}=\sum_{i=1}^n x_i y_i^q.$$ 
We say that $\bf x$ is Hermitian orthogonal to $\bf y$ if ${\bf x}*{\bf y}=0.$
For $E\subset \F_{q^2}^n$ we denote
$$\overline{E}:=\{(x_1^q,\hdots,x_n^q)|(x_1,\hdots,x_n)\in E\}.$$
For ${\bf x}=(x_1,\hdots,x_n)\in \F_{q^2}^n,$ we denote ${\bf x}^q=(x_1^q,\hdots,x_n^q).$
The Euclidean (resp. Hermitian)  {\em dual} of $C$,
denoted by $C^{\perp_E}$ (resp. $C^{\perp_H}$) is the set of vectors orthogonal to every
codeword of $C$ under the Euclidean (resp. Hermitian) inner product. A linear code
$C$ is called {\em Euclidean (resp. Hermitian) self-orthogonal} if $C\subset C^{\perp_E}$ (resp. $C\subset C^{\perp_H}$). A code $C$ is called 
{\em Euclidean (resp. Hermitian) self-dual} if $C=C^{\perp_E}$ (resp. $C=C^{\perp_H}$). 

It is easy to verify that $C^{\perp_H}=\overline{C^{\perp_E}}.$
It is well known that a self-dual code can only exist for even length. If $C$ is an $[n,k,d]$ code, then from the Singleton bound, its minimum distance is bounded by
$$d(C)\le n-k+1.$$
A code meeting the above bound is called {\em Maximum Distance Separable} ({MDS}) code and is called almost MDS if its minimum distance is one unit less than the MDS case. A code is called {\it optimal} if it has the highest possible minimum distance for its length and dimension and thus an MDS code is optimal.

\medskip
\subsection{Unitary group over finite fields}
As we will see in the next sections, a unitary group over finite fields plays a key role in our construction of Hermitian self-dual codes. 

The {\em unitary group} of index $n$ over a finite field $q^2$ elements with $q=p^{m}$  is defined by
$${\cal U}_n(q^2):=\{A\in GL(n,q^2)| A{\overline A^\top}=I_n\},$$
where $\overline{A}$ is the matrix obtained from $A$ by taking the conjugate of all entries of A, that is, 
if $A=(a_{i,j})_{1\le i,j\le n}$ then $\overline {A}=(a_{i,j}^{q})_{1\le i,j\le n}.$ 

The order of the group was determined by Wall \cite{Wall} and is given as follows

\begin{equation}
|{\cal U}_n(q^2)|=q^{\frac{n^2-n}{2}}\prod\limits_{i=1}^n(q^i-(-1)^i).
\label{eq:order-group}
\end{equation}
In what follows, we present some elements used to generate a unitary group. 
Let $q=p^{m}$ for some prime $p$ and some positive integer $m.$ Let $\theta=\frac{p-1}{2}\in {\mathbb F}_p$ if $p\neq 2$ and $\theta=1$ otherwise.
Let $a,b,c,d\in {\mathbb F}_{q^2}$ such that 
\begin{equation}
\begin{cases}
a^{{q}+1}+b^{{q}+1}=1,\\ 
c^{{q}+1}+d^{{q}+1}=1,\\
a^{{q}}c+b^{{q}}d=0, \\
\end{cases}\label{eq:transvection}
\end{equation}
and ${\bf v}=(a-1){\bf b}_1+b {\bf b}_2,{\bf w}=c{\bf b}_1+(d-1) {\bf b}_2$, ${\bf u}={\bf b}_1+{\bf b}_2+{\bf b}_3+{\bf b}_4$ if $n\ge 4$, where $B=\{ {\bf b}_1,\hdots,{\bf b}_n \}$ is the canonical basis of $\mathbb{F}_{q^2}^n$. 

Define two linear maps
\begin{equation}\label{eq:linearmaps}
\begin{array}{llll}
T_{{\bf u},\theta}: &\mathbb{F}_{q^2}^n \longrightarrow \mathbb{F}_{q^2}^n,& T_{a,b,c,d}:& \mathbb{F}_{q^2}^n \longrightarrow \mathbb{F}_{q^2}^n\\
 &{\bf x} \mapsto {\bf x} +\theta({\bf x\cdot u}){\bf u}& &{\bf x} \mapsto {\bf x}+({\bf x \cdot {\bf v}}){\bf b}_1+({\bf x}\cdot {\bf w}){\bf b}_2.\\
\end{array}
\end{equation}
Denote
$${\cal T}_n{(q^2)}:=
\begin{cases}
\langle {\cal P}_n ,T_{a,b,c,d}\rangle\text{ if } n\leq 3,\\
\langle {\cal P}_n ,T_{a,b,c,d},T_{{\bf u},\theta}\rangle, \text{ otherwise},
\end{cases}$$
where ${\cal P}_n$ is the permutation group of $n$ elements.

\begin{lem} Let $q=p^m.$ There exist solutions $a,b,c,d\in \F_{q^2}$ for the system of equations defined by
\begin{equation}
\begin{cases}
a^{{q}+1}+b^{{q}+1}=1,\\ 
c^{{q}+1}+d^{{q}+1}=1,\\
a^{{q}}c+b^{{q}}d=0, \\
\end{cases}
\end{equation}
\end{lem}
\begin{proof} Since we can choose $c=-b$ and $d=a,$ it is enough to prove that there exist $a,b\in \F_{q^2}$ such that
\begin{equation}
\begin{cases}
a^{{q}+1}+b^{{q}+1}=1\\
a^{{q}}(-b)+b^{{q}}a=0\\
\end{cases}\label{eq:transvection2}
\end{equation}
There exist $a,b\in \F_{q^2}$ such that the first equation of the system (\ref{eq:transvection2}) holds. Now we have that $a^{q}=a,b^{q}=b$ and thus the second equation of the system also holds.
\end{proof}
\begin{rem} 
It should be noted that $(a,b)=(1,0),(c,d)=(0,1)$ are solutions to the system (\ref{eq:transvection}). However with these values, the linear map $T_{a,b,c,d}$ is just the identity map and thus it is out of interest. 
\end{rem}
\begin{prop} The group ${\cal T}_n{(q^2)}$ is a subgroup of ${\cal U}_n(q^2).$
\end{prop}
\begin{proof} Obviously for any $A\in  {\cal P}_n$, we have $A=\overline{A}$ and hence $A$ is  Hermitian orthogonal. We also have
$T_{{\bf u},\theta}=\overline{T_{{\bf u},\theta}}$ since it is a matrix with all entries in $\F_p$ and thus $T_{{\bf u},\theta}$ is Hermitian orthogonal . Finally under the conditions given in Eq. (\ref{eq:transvection}) we have that $T_{a,b,c,d}$ satisfies 
$T_{a,b,c,d}(b_i)*T_{a,b,c,d}(b_j)=\delta_{ij}$, where $\delta_{ij}$ is the Kronecker symbol, and it is thus Hermitian orthogonal.
\end{proof}

The orders of the subgroup ${\cal T}_n(q^2)$ are computed using Magma \cite{mag} and given in Table \ref{Table:order} as well as compared with the orders of the unitary group given in Eq. (\ref{eq:order-group}). 
\begin{table}[h]\caption{The orders of ${\cal T}_n(q^2)$ computed in Magma \cite{mag} compared with those of  ${\cal U}_n(q^2)$ Eq. (\ref{eq:order-group}) for $n=3,4$.}
$$
\begin{array}{c|c|c|c|c}

q^2&|{\cal T}_3(q^2)|&|{\cal U}_3(q^2)|&|{\cal T}_4(q^2)|&|{\cal U}_4(q^2)| \\
\hline

  3^2  & 24192& 24192 & 52254720& 52254720 \\

    4^2& 312000& 312000 &5091840000& 5091840000 \\    

   5^2 & 2268000& 2268000 & 176904000000& 176904000000 \\

   7^2 & 45308928& 45308928 &37298309529600& 37298309529600 
\end{array}
$$\label{Table:order}
\end{table}

We have already seen that the order of unitary groups in Table \ref{Table:order} grows very fast when the dimension becomes larger. The memory space for storing such matrices will be a challenging problem and visiting all the elements in the group is not possible say for $n\ge 4$. To search for optimal Hermitian self-dual codes in the next section we need the matrices with each row having as many non-zero entries as possible. For that, we propose the following algorithm:\\
\noindent \rule[-.1cm]{\linewidth}{0.2mm}
{\bf Algorithm 1}\\
\noindent \rule[.3cm]{\linewidth}{0.2mm}

\vspace{-.3cm}

\noindent\textbf{Input:} $n\ge 4,m,s$: positive integers\\
\textbf{Output:} A list $S$ of $n\times n$ unitary matrices

\vspace{.2cm}
\noindent
\begin{small}
\begin{enumerate}
\item $T_{0}:=$ transposition; 
\item $T_{1}:=$ cycle of length $n$;
\item $T_{u,\theta}:=$ linear map in Eq. (\ref{eq:linearmaps});
\item $T_{a,b,c,d}:=$ linear map in Eq. (\ref{eq:linearmaps});
\item $N:=(T_{u,\theta}T_1T_{a,b,c,d}T_0)^m$;
\item $P:=(T_{a,b,c,d}T_0T_{u,\theta}T_1)^m$;
\item $Q:=(T_{u,\theta}T_0T_{a,b,c,d}T_1)^m$;
\item $R:=(T_{a,b,c,d}T_1T_{u,\theta}T_0)^m$;
\item $S:=[~];$
\item  for $i:=0$ to $s$ do
\item  for $j:=0$ to $s$ do
\item  for $k:=0$ to $s$ do
\item  for $l:=0$ to $s$ do
\item $L:=N^iP^jQ^kR^l$;
\item {\tt Append}$(S,L)$;
\item end for;
\item end for;
\item end for;
\item end for;
\item return $S$;
\end{enumerate}
\end{small}
\noindent \rule[.3cm]{\linewidth}{0.2mm}
\section{Construction of Hermitian self-dual codes}\label{section:con}
In this section, we introduce some constructions of Hermitian self-dual codes based on the elements in the unitary group. 

First let us recall the classical constructions of Hermitian self-dual codes  provided by Gaborit \cite{Gabo}, which are known to be the pure and bordered quadratic double circulant construction whose generator matrices of the code are of the following forms
\begin{equation*}
G_r=(I_r|Q_r)
\end{equation*} and 
\begin{equation*} \label{eq:circulant}
G_{r}=\left(\begin{array}{c|ccc|c|ccc}
1&0&\cdots&0&\alpha&\beta&\cdots &\beta\\
\hline
0&&&&\gamma&&&\\
\vdots&&I_r&&\vdots&&Q_r&\\
0&&&&\gamma&&&\\
\end{array}\right)
\end{equation*}
respectively, where $\alpha,\beta,\gamma\in \F_{q^2}$ and $Q_q$ is a $r\times r$ circulant matrix indexed by quadratic residues in $\F_{q^2}$.

There is also a recursive construction of Hermitian self-dual codes given by Kim et al. \cite{KimLee} as follows.

\begin{prop}([building-up]\cite{KimLee})
\label{kimthm} Let $a$ be in $\F_{q^2}$ such that $a^{q+1}=-1$ in $\F_{q^2}$. Let $G_0$ be a generator matrix (not necessarily in standard form) of a self-dual code $C_0$ over $\F_{q^2}$ of length $2n$, where ${\bf g}_i$ are the rows of the matrices $G_0$, for $1\le i \le n$. Let ${\bf x}=(x_1,\hdots,x_n,x_{n+1},\hdots,x_{2n})$ be a vector in $\F_{q^2}^{2n}$ with ${\bf x* x}=-1$ in $\F_{q^2}$. Suppose that $y_i:={\bf x}*{\bf g}_i$ for $1\le i \le n$. Then the following matrix:

\begin{equation*} \label{eq:build-up}
G = \left( \begin{array}{cc|cccccc}
  1   & 0   & &  &  & {\bf{x}}   &  &  \\ \hline
  -y_1 & ay_1 &    &        &     &         &        &        \\
  \vdots & \vdots &  &  & & G_0  & & \\
 -y_{n} & ay_{n} &    &        &     & \\
\end{array}
\right)
\end{equation*}
generates a Hermitian self-dual code $C$ over $\F_{q^2}$ of length $2n+2.$
\end{prop}

It should be noted that if $\omega$ is a primitive root of $\F_{q^{2}}$ and $\alpha=\omega^{\frac{q-1}{2}}$, then we have $\alpha^{q+1}=-1.$
Thus we derive the construction of Hermitian self-dual codes over ${\mathbb F}_{q^2}$ as follows.

\begin{prop}\label{prop:sd1}Let $q=p^m$ with $p$ being a prime.
 Let $L \in {\cal U}_{n}(q^2)$ and fix  $\alpha \in {\mathbb F}_{q^2}$ such that $\alpha^{q+1}=-1.$ Then the matrix $G_{n}$ of the following form:
\begin{equation}\label{eq:sd1}
G_{n}=\left(
\begin{array}{c|c}
L^\top & \alpha L\\
\end{array}
\right),
\end{equation}
 generates a $q^2-$ary self-dual $[2n,n]$ code.
\end{prop}
\begin{proof} First note that if $L$ is in ${\cal U}_n(q^2)$ then so is $L^\top$. Let $g_i$ be the $i-$th row of $G_n$. It follows that $g_i*{g_j}=0$ for $1\le i,j\le n,$ which means that the code having $G_n$ as its generator matrix is Hermitian self-orthogonal. Since the dimension of the code is equal to the row rank of $L$, the result follows.
\end{proof}



Similar to Construction (\ref{eq:sd1}), we have the following. 
\begin{prop}\label{prop:sd2}Let $q=p^m$ with $p$ being a prime.
 Let $L \in {\cal U}_{n}({q^2})$ and fix  $\alpha \in {\mathbb F}_q$ such that $\alpha^{q+1}= -1.$ Then the matrix $G_{n}$ of the following form:
\begin{equation}\label{eq:sd2}
G_{n}=\left(
\begin{array}{c|c}
I_n & \alpha L\\
\end{array}
\right),
\end{equation}
 generates a $q^2-$ary Hermitian self-dual $[2n,n]$ code.
\end{prop}
\begin{proof} Let $g_i$ be the $i-$th row of $G_n$. It follows that $g_i*{g_j}=0$ for $1\le i,j\le n,$ which means that the code having $G_n$ as its generator matrix is Hermitian self-orthogonal. Since the dimension of the code is equal to $n$, the result follows.
\end{proof}

In the rest of the paper, $J_n$ denotes the $n\times n$ matrix with all entries equal to $1.$

\begin{prop}  Let $n\equiv 2\pmod p$ and $q=p^m.$ Fix $a\in {\mathbb F}_q$ such that $a^{q+1}= -1.$ Then for any $L{\in \cal U}_n({q^2})$, a code with the following generator matrix is a $q^2-$ary Hermitian self-dual $[2n,n]$ code:

\begin{equation} \label{eq:sd3}
G_{n}=\left(\begin{array}{c|c}
&\\
J_n-I_n&a L\\
&\\
\end{array}\right).
\end{equation}
\end{prop}
\begin{proof} Let $g_i=(l_i|r_i)$ be the $i-$th row of $G_n$. It follows that $g_i* g_j=l_i* l_j+r_i* r_j=n-2\equiv 0\pmod p$ for $1\le i,j\le n,$ which means that the code having $G_n$ as its generator matrix is Hermitian self-orthogonal. Since the dimension of the code is equal to the row rank of $L$, the result follows.
\end{proof}

Similarly we have the following construction.

\begin{prop}  Let $n\equiv -2\pmod p$ and $q=p^m.$ Fix $a\in {\mathbb F}_q$ such that $a^{q+1}= -1.$ Then for any $L{\in \cal U}_n(q^2)$, a code with the following generator matrix is a $q^2-$ary Hermitian self-dual $[2n,n]$ code:

\begin{equation} \label{eq:sd4}
G_{n}=\left(\begin{array}{c|c}
&\\
J_n+I_n&a L\\
&\\
\end{array}\right).
\end{equation}
\end{prop}

\begin{lem}\label{lem:double1} Let $q=p^m$ and $a\in \F_{q^2}$ such that $a^{q+1}=-1.$ Let $L\in {\cal U}_n(q^2)$ and $L_i$ denote its $i-$th row. Assume there exist $\alpha,\beta,\gamma,\delta,\lambda,\theta$ in $\F_{q^2}$ satisfying 
\begin{equation}\label{system:1}
\begin{cases}
\delta^{q+1}+(n-2)+\gamma^{q+1}=0\\
\theta^{q+1}+n\beta^{q+1}+\alpha^{q+1}+n\lambda^{q+1}=0\\
\theta\delta^{q}+(n-1)\beta+\alpha\gamma^{q}+\lambda a^{q}=0.\\
\end{cases}
\end{equation}
Then the code with the following generator matrix
\begin{equation} \label{eq:double1}
G_{n}=\left(\begin{array}{c|ccc|c|ccc}
\theta&\beta&\cdots&\beta&\alpha&&\lambda(L_1+\cdots+L_n)&\\
\hline
\delta&&&&\gamma&&&\\
\vdots&&J_n-I_n&&\vdots&&aL&\\
\delta&&&&\gamma&&&\\
\end{array}\right)
\end{equation}
is a $q^2-$ary Hermitian self-orthogonal $[2n+2,\ge n]$ code. Moreover if $\delta=0$ and $\theta\not =0$ then $G_n$ generates a $q^2-$ary Hermitian self-dual $[2n+2,n+1]$ code.
\end{lem}
\begin{proof} Let $C$ be the code generated by $G_n$. For $1\le i \le n+1$, let $g_{i}$ be the $i-$th row of $G_n.$ A simple calculation implies that
\begin{equation*}
\begin{cases}
g_i*g_i=\delta^{q+1}+(n-1)+\gamma^{q+1}+a^{q+1}\text{ for } i\le 2 \le n+1\\
g_i*g_j=\delta^{q+1}+(n-2)+\gamma^{q+1}\text{ for } 2\le i\not= j\le n+1\\
g_1*g_1=\theta^{q+1}+n\beta^{q+1}+\alpha^{q+1}+n\lambda^{q+1}\\
g_1*g_i=\theta\delta^{q}+(n-1)\beta+\alpha\gamma^{q}+\lambda a^{q}\text{ for } i\le 2 \le n+1.\\
\end{cases}
\end{equation*}
For $C$ to be Hermitian self-orthogonal, we have to take $a^{q+1}=-1$ and the system (\ref{system:1}).
Now since $L$ is invertible, the row rank of the matrix $G_n$ is at least $n$. If $\delta=0$ and $\theta\not=0$ then by applying elementary row operations on the last $n$ rows of $G_n$ followed by swapping the first column and the $(n+2)-$th column, it is easy to see that the row rank of $G_n$ is exactly $n+1$ and thus the code $C$ is $q^2-$ary Hermitian self-dual.
\end{proof}
We deduce the constructions of $q^2-$ary Hermitian self-dual codes from the matrix $G_n$ in Eq. (\ref{eq:double1}) as follows.

\begin{thm}\label{thm:double1} Assume that $q=p^m,a^{q+1}=-1,\delta=0,\gamma^{q+1}=2-n\not \equiv 0 \pmod p .$ Then
\begin{enumerate}
\item for $0\not=\theta\in \F_{q^2},\beta=a\theta\gamma,\alpha=\frac{(1-n)\beta}{\gamma^q},\lambda=0,$ the matrix $G_n$ generates a $q^2-$ary Hermitian self-dual $[2n+2,n+1]$ code.
\item for $\theta=\sqrt[q+1]{n},\beta=\frac{1}{n-1+a\theta\gamma^q},\alpha=a\theta\beta,\lambda=a,$ the matrix $G_n$ generates a $q^2-$ary Hermitian self-dual $[2n+2,n+1]$ code.
\item for $0\not=\theta\in \F_{q^2},\beta=\frac{a\theta \gamma^q}{2-n},\alpha=a\theta,\lambda=\beta a,$ the matrix $G_n$ generates a $q^2-$ary Hermitian self-dual $[2n+2,n+1]$ code.

\end{enumerate}
\end{thm}
\begin{proof}  It is obvious that $\delta=0$ and $\gamma^{q+1}=2-n$ satisfy the first equation of system (\ref{system:1}).
\begin{enumerate}
\item Plugging $\lambda=0$ in system (\ref{system:1}) and letting $\theta$ be arbitrary, we obtain two equations in two varaibles $\beta,\alpha$ having the solutions in the desired form.
\item Plugging $\lambda=a$ in system (\ref{system:1}) and letting $\theta=\sqrt[q+1]{n}$, we obtain two equations in two varaibles $\beta,\alpha$ having the solutions in the desired form. If $a\theta\gamma^q=1-n$, then by raising both sides to the power $q+1$, we get that $a^{q+1}\theta^{q+1}\gamma^{q+1}=-n(2-n)=n^2-2n+1,$ which is a contradiction.
 \item Plugging $\lambda=a\beta$ in system (\ref{system:1}) and letting $\theta$ be arbitrary, we obtain two equations in two varaibles $\beta,\alpha$ having the solutions in the desired form.

\end{enumerate}
\end{proof}
\begin{rem}
\begin{enumerate}
\item It should be noted that if $\delta=0$ and $\gamma=0$ then the code with the generator matrix (\ref{eq:double1}) is still $q^2-$ary Hermitian self-dual but it has minimum distance at most $2$.
\item For 1) and 3), the parameters of the constructed codes are the same for any $\theta\not=0$ so we can choose $\theta=1.$
\item Taking $\theta=1,\delta=0,\lambda=0$ in Theorem \ref{thm:double1} 1), we can express Equation (\ref{eq:double1}) as 
\begin{equation*} 
G_{n}=\left(\begin{array}{c|ccc|c|ccc}
1&\beta&\cdots&\beta&\alpha&&0\cdots 0&\\
\hline
0&&&&\gamma&&&\\
\vdots&&J_n-I_n&&\vdots&&aL&\\
0&&&&\gamma&&&\\
\end{array}\right)
\end{equation*}
By applying elementary row operations on $G_n$, we obtain
\begin{equation*}
G'_{n}=\left(\begin{array}{c|ccc|c|ccc}
1&\beta&\cdots&\beta&\alpha&&0\cdots 0&\\
\hline
0&&&&\gamma_1&&&\\
\vdots&&Q&&\vdots&&I_n&\\
0&&&&\gamma_n&&&\\
\end{array}\right),
\end{equation*}
which can be viewed as a generalized construction of the quadratic double circulant construction.
\end{enumerate}
\end{rem}

\begin{lem} \label{lem:dependent}Assume that $\delta=1.$ Let $g_i$ denote the $i-$th row of $G_n$ of Eq. (\ref{eq:double1}). If $g_1$ is in span$(g_2,\hdots,g_{n+1})$ then $\theta=\epsilon n$ and $\beta=\epsilon (n-1)$ for some non-zero $\epsilon \in \F_{q^2}.$
\end{lem}
\begin{proof} Assume that $g_1=\lambda_1g_2+\cdots+\lambda_ng_{n+1}$ for some $\lambda_1,\hdots,\lambda_n\in \F_{q}.$ Considering the first $n+1$ equations defined by this system, we get that $\theta=\sum\limits_{i=1}^n{\lambda_i},\beta=\sum\limits_{i=2}^n{\lambda_i}=\cdots=\sum\limits_{i=1}^{n-1}{\lambda_i}.$ Thus $\lambda_1=\cdots=\lambda_n$ and the result follows.

\end{proof}
\begin{thm}\label{thm:double2} Assume that $q=p^m,a^{q+1}=-1,\delta=1,\gamma^{q+1}=1-n\not \equiv 0\pmod p.$ Then
\begin{enumerate}
\item for $p\equiv 1\pmod 2, n\not \equiv 3 \pmod p,\theta\in \F_{p},\beta=2,\alpha=\frac{2(1-n)-\theta}{\gamma^q},\lambda=0,(\theta-2)^2n-(2\theta^2-4\theta+4)\equiv 0\pmod p,$ the matrix $G_n$ generates a $q^2-$ary Hermitian self-dual $[2n+2,n+1]$ code.

\item for $ n\not \equiv 2 \pmod p,\theta\in \F_{p},\beta=1,\alpha=\frac{(1-n)-\theta}{\gamma^q},\lambda=0,(\theta-1)^2n-(2\theta^2-2\theta+1)\equiv 0\pmod p,$ the matrix $G_n$ generates a $q^2-$ary Hermitian self-dual $[2n+2,n+1]$ code.

\item for $ n\not \equiv 0\pmod p, \theta=\sqrt[q+1]{n},\beta=\frac{1-\theta}{n-1+a\theta\gamma^q},\alpha=a\theta\beta,\lambda=a,$ the matrix $G_n$ generates a $q^2-$ary Hermitian self-dual $[2n+2,n+1]$ code.
\item for $n\not \equiv 2\pmod p, \theta\in \F_{q^2},\beta=\frac{\theta (1+a\gamma^q)}{2-n},\alpha=a\theta,\lambda=\beta a,$ the matrix $G_n$ generates a Hermitian self-dual $[2n+2,n+1]$ code.

\end{enumerate}
\end{thm}

\begin{proof} It is obvious that $\delta=1$ and $\gamma^{q+1}=1-n$ satisfy the first equation of system (\ref{system:1}).
\begin{enumerate}
\item Plugging $\lambda=0$ in system (\ref{system:1}) and letting $\theta$ be arbitrary, we obtain two equations in two varaibles $\beta,\alpha$ having the solutions in the desired form.
 From Lemma \ref{lem:dependent}, if $g_1$ is in $\text{span}(g_2,\hdots,g_{n+1})$ then we can also write $\theta=n$ and $\beta=(n-1)$. Since $\beta=2$, we get $n-1\equiv 2 \pmod p$ which is a contradiction. Hence the code dimension is $n+1$ and the result follows.
\item The result follows from the same computation and reasoning as 1).
\item Plugging $\lambda=a$ in system (\ref{system:1}) and letting $\theta=\sqrt[q+1]{n}$, we obtain two equations in two varaibles $\beta,\alpha$ having the solutions in the desired form. If $a\theta\gamma^q=1-n$, then by raising both sides to the power $q+1$, we get that $a^{q+1}\theta^{q+1}\gamma^{q+1}=-n(1-n)=n^2-2n+1,$ that is $(1-n)\equiv 0 \pmod p$, which is a contradiction. We now prove that the code dimension is $n+1$. From Lemma \ref{lem:dependent}, if $g_1$ is in $\text{span}(g_2,\hdots,g_{n+1})$ then we can also write $\theta=n$ and $\beta=(n-1)$. Since $\theta=\sqrt[q+1]{n}$, we get $\theta^{q+1}=n=n^{q+1}=n^2$ and thus $n\equiv 0,1\pmod p$ which is a contradiction.
 \item Plugging $\lambda=a\beta$ in system (\ref{system:1}) and letting $\theta$ be arbitrary, we obtain two equations in two varaibles $\beta,\alpha$ having the solutions in the desired form. It remains to prove that the code dimension is $n+1$. First note that if $a\gamma^q=-1$ then by raising both sides to power $q$, we get $a^q\gamma=-1$ and hence $a^{q+1}\gamma^{q+1}=-(1-n)=1$, which is a contradiction to the hypothesis. If $\frac{\theta}{\beta}=\frac{n}{n-1}$ then from $\beta=\frac{\theta (1+a\gamma^q)}{2-n}$, we get

$$\frac{n}{n-1}=\frac{2-n}{1+a\gamma^q}$$ and by raising both sides to the power $q$, we also obtain
$$\frac{n}{n-1}=\frac{2-n}{1+a^q\gamma}.$$
Matching the two equations together gives $a\gamma^q=a^q\gamma$ and hence $\left(\frac{\gamma}{a}\right)^{q-1}=1$ which means that $\frac{\gamma}{a}\in \F_{q}.$ Since $\left(\frac{\gamma}{a}\right)^{q+1}=n-1,$ we get $\left(\frac{\gamma}{a}\right)^{q+1}\left(\frac{\gamma}{a}\right)^{q-1}=\left(\frac{\gamma}{a}\right)^2=n-1$ and thus $\gamma^2=(n-1)a^2.$ Now since $\gamma^{2t}=a^{2t}$ for some positive integer $t$, we obtain $a^{2t}=(n-1)^ta^{2t}$ and it implies that $n-1=1,$ which is a constradiction. Hence 
$\frac{\theta}{\beta}\not=\frac{n}{n-1}$ and the result follows by Lemma \ref{lem:dependent}.

\end{enumerate}
\end{proof}

\begin{rem} 
\begin{enumerate}
\item When $n\equiv 0,1 \pmod p$, the constructed codes in Theorem \ref{thm:double2} 2) are not $q^2-$ary Hermitian self-dual since the first row is in the spanned space of the other $n$ rows, more precisely $g_1=g_2+\cdots+g_{n+1}$.
\item When $n-1\equiv 0\pmod p$,  the construction  in Theorem \ref{thm:double2} 3) still make sense and by exchanging the first and $(n+2)-$th columns,  it is equivalent to that in Theorem \ref{thm:double1} 3).
\end{enumerate}
\end{rem}

\begin{cor}
Assume that $p\equiv 1 \pmod 2,$ $n\equiv 2 \pmod p$ and $n\not \equiv 1,3 \pmod p.$ Then for $a^{q+1}=-1,\delta=1,\gamma^{q+1}=1-n$, $\theta=1,\beta=2,\alpha=\frac{1-2n}{\gamma^q}, \lambda=0,$ the matrix $G_n$ generates a $q^2-$ary Hermitian self-dual $[2n+2,n+1]$ code.
\end{cor}
\begin{proof} The result follows by plugging $\theta=1$ in Theorem \ref{thm:double2} 1).
\end{proof}
\begin{cor}
Assume that $n\equiv 5 \pmod p$ and $n\not \equiv 1,2 \pmod p.$ Then for $a^{q+1}=-1,\delta=1,\gamma^{q+1}=1-n$, $ \theta=2,\beta=1,\alpha=\frac{-1-n}{\gamma^q},\lambda=0,$ the matrix $G_n$ generates a $q^2-$ary Hermitian self-dual $[2n+2,n+1]$ code.
\end{cor}
\begin{proof} The result follows by plugging $\theta=2$ in Theorem \ref{thm:double2} 1).
\end{proof}


Similar to Lemma \ref{lem:double1} we have the following.
\begin{lem}\label{lem:double2} Let $q=p^m$ and $a\in \F_{q^2}$ such that $a^{q+1}=-1.$ Let $L\in {\cal U}_n(q^2)$ and $L_i$ denote its $i-$th row. Assume there exist $\alpha,\beta,\gamma,\delta,\lambda,\theta$ satisfying 
\begin{equation}\label{eq:system2}
\begin{cases}
\delta^{q+1}+(n+2)+\gamma^{q+1}=0\\
\theta^{q+1}+n\beta^{q+1}+\alpha^{q+1}+n\lambda^{q+1}=0\\
\theta\delta^{q}+(n+1)\beta+\alpha\gamma^{q}+\lambda a^{q}=0.\\
\end{cases}
\end{equation}
Then the code with the following generator matrix
\begin{equation} 
G_{n}=\left(\begin{array}{c|ccc|c|ccc}
\theta&\beta&\cdots&\beta&\alpha&&\lambda(L_1+\cdots+L_n)&\\
\hline
\delta&&&&\gamma&&&\\
\vdots&&J_n+I_n&&\vdots&&aL&\\
\delta&&&&\gamma&&&\\
\end{array}\right)
\label{eq:double2}
\end{equation}
is a $q^2-$ary Hermitian self-orthogonal $[2n+2,\ge n]$ code.  Moreover if $\delta=0$ and $\theta\not =0$ then $G_n$ generates a $q^2-$ary Hermitian self-dual $[2n+2, n+1]$ code.
\end{lem}

\begin{proof} Let $C$ be the code generated by $G_n$. For $1\le i \le n+1$, let $g_{i}$ be the $i-$th row of $G_n.$ A simple calculation implies that
\begin{equation*}
\begin{cases}
g_i*g_i=\delta^{q+1}+(n+3)+\gamma^{q+1}+a^{q+1}\text{ for } i\le 2 \le n+1\\
g_i*g_j=\delta^{q+1}+(n+2)+\gamma^{q+1}\text{ for } 2\le i\not= j\le n+1\\
g_1*g_1=\theta^{q+1}+n\beta^{q+1}+\alpha^{q+1}+n\lambda^{q+1}\\
g_1*g_i=\theta\delta^{q}+(n+1)\beta+\alpha\gamma^{q}+\lambda a^{q}\text{ for } i\le 2 \le n+1.\\
\end{cases}
\end{equation*}
The rest follows from the same reasoning as that in Lemma \ref{lem:double1}
\end{proof}

We deduce the constructions of $q^2-$ary Hermitian self-dual codes from the matrix $G_n$ in Eq. (\ref{eq:double2}) as follows.

\begin{thm}\label{thm:double3} Assume that $q=p^m,a^{q+1}=-1,\delta=0,\gamma^{q+1}=-2-n\not \equiv 0 \pmod p .$ Then
\begin{enumerate}
\item for $0\not=\theta\in \F_{q^2},\beta=a\theta\gamma,\alpha=\frac{-(n+1)\beta}{\gamma^q},\lambda=0,$ the matrix $G_n$ generates a $q^2-$ary Hermitian self-dual $[2n+2,n+1]$ code.
\item for $\theta=\sqrt[q+1]{n},\beta=\frac{1}{1+n+a\theta\gamma^q},\alpha=a\theta\beta,\lambda=a,$ the matrix $G_n$ generates a $q^2-$ary Hermitian self-dual $[2n+2,n+1]$ code.
\item for $n\not \equiv 0\pmod p,0\not=\theta\in \F_{q^2},\beta=\frac{-a\theta \gamma^q}{n},\alpha=a\theta,\lambda=\beta a,$ the matrix $G_n$ generates a $q^2-$ary Hermitian self-dual $[2n+2,n+1]$ code.

\end{enumerate}
\end{thm}

\begin{proof} The proof follows from the same reasoning as that in Theorem \ref{thm:double1}.
\end{proof}



\begin{lem} \label{lem:dependent2}Assume that $\delta=1.$ Let $g_i$ denote the $i-$th row of $G_n$ of Eq. (\ref{eq:double1}). If $g_1$ is in span$(g_2,\hdots,g_{n+1})$ then $\theta=\epsilon n$ and $\beta=\epsilon (n+1)$ for some non-zero $a\epsilon \in \F_{q^2}.$
\end{lem}
\begin{proof} Assume that $g_1=\lambda_1g_2+\cdots+\lambda_ng_{n+1}$ for some $\lambda_1,\hdots,\lambda_n\in \F_{q}.$ Considering the first $n+1$ equations defined by this system, we get that $\theta=\sum\limits_{i=1}^n{\lambda_i},\beta=\lambda_1+\sum\limits_{i=1}^n{\lambda_i}=\cdots=\lambda_n+\sum\limits_{i=1}^{n}{\lambda_i}.$ Thus $\lambda_1=\cdots=\lambda_n$ and the result follows.
\end{proof}

\begin{thm}\label{thm:double4} Assume that $q=p^m,a^{q+1}=-1,\delta=1,\gamma^{q+1}=-3-n\not \equiv 0 \pmod p .$ Then
\begin{enumerate}
\item for $p\equiv 1\pmod 2, n\not \equiv 1 \pmod p,\theta\in \F_{p},\beta=2,\alpha=\frac{-2(1+n)-\theta}{\gamma^q},\lambda=0,(\theta-2)^2n+2\theta^2-4\theta-4\equiv 0\pmod p,$ the matrix $G_n$ generates a $q^2-$ary Hermitian self-dual $[2n+2,n+1]$ code.
 \item for $ n\not \equiv 0 \pmod p,\theta\in \F_{p},\beta=1,\alpha=\frac{-(1+n)-\theta}{\gamma^q},\lambda=0,(\theta-1)^2n+2\theta^2-2\theta-1\equiv 0\pmod p,$ the matrix $G_n$ generates a $q^2-$ary Hermitian self-dual $[2n+2,n+1]$ code.
\item for $n\not \equiv 0, 1\pmod p,\theta=\sqrt[q+1]{n},\beta=\frac{1-\theta}{n+1+a\theta\gamma^q},\alpha=a\theta\beta,\lambda=a,$ the matrix $G_n$ generates a $q^2-$ary Hermitian self-dual $[2n+2,n+1]$ code.
\item for $n\not \equiv 0,-2\pmod p, \theta\in \F_{q^2},\beta=\frac{-\theta (1+a\gamma^q)}{n},\alpha=a\theta,\lambda=\beta a,$ the matrix $G_n$ generates a $q^2-$ary Hermitian self-dual $[2n+2,n+1]$ code.

\end{enumerate}
\end{thm} 

\begin{proof} It is obvious that $\delta=1$ and $\gamma^{q+1}=-3-n$ satisfy the first equation of system (\ref{eq:system2}).
\begin{enumerate}
\item Plugging $\lambda=0$ in system (\ref{eq:system2}) and letting $\theta$ be arbitrary, we obtain two equations in two varaibles $\beta,\alpha$ having the solutions in the desired form.
 From Lemma \ref{lem:dependent2}, if $g_1$ is in $\text{span}(g_2,\hdots,g_{n+1})$ then we can also write $\theta=n$ and $\beta=(n+1)$. Since $\beta=2$, we get $n+1\equiv 2 \pmod p$ which is a contradiction.
\item The result follows from the same computation and reasoning as 1).
\item Plugging $\lambda=a$ in system (\ref{eq:system2}) and letting $\theta=\sqrt[q+1]{n}$, we obtain two equations in two varaibles $\beta,\alpha$ having the solutions in the desired form. If $a\theta\gamma^q=1-n$, then by raising both sides to the power $q+1$, we get that $a^{q+1}\theta^{q+1}\gamma^{q+1}=-n(-3-n)=n^2+2n+1,$ that is $(n-1)\equiv 0 \pmod p$, which is a contradiction. From Lemma \ref{lem:dependent2}, if $g_1$ is in $\text{span}(g_2,\hdots,g_{n+1})$ then we can also write $\theta=n$ and $\beta=(n+1)$. Since $\theta=\sqrt[q+1]{n}$, we get $\theta^{q+1}=n=n^{q+1}=n^2$ and thus $n\equiv 0,1\pmod p$ which is a contradiction.
 \item Plugging $\lambda=a\beta$ in system (\ref{eq:system2}) and letting $\theta$ be arbitrary, we obtain two equations in two varaibles $\beta,\alpha$ having the solutions in the desired form. If $\frac{\beta}{\theta}=\frac{n+1}{n}$ then from $\beta=\frac{-\theta (1+a\gamma^q)}{n}$, we get

$$n+1=-(1+a\gamma^q),(n+1)^q=-({1+a^q\gamma}).$$
Matching the two equations together gives $a\gamma^q=a^q\gamma$ and hence $\left(\frac{\gamma}{a}\right)^{q-1}=1.$ Since $\left(\frac{\gamma}{a}\right)^{q+1}=n+3,$ we get $\left(\frac{\gamma}{a}\right)^{q+1}\left(\frac{\gamma}{a}\right)^{q-1}=\left(\frac{\gamma}{a}\right)^2=n+3$ and thus $\gamma^2=(n+3)a^2.$ Now since $\gamma^{2t}=a^{2t}$ for some positive integer $t$, we obtain $a^{2t}=(n+3)^ta^{2t}$ and it implies that $n+3=1,$ which is a constradiction. Hence 
$\frac{\beta}{\theta}\not=\frac{n+1}{n}$ and the result follows by Lemma \ref{lem:dependent2}.

\end{enumerate}
\end{proof}

\begin{cor}
Assume that $p\equiv 1 \pmod 2,$ $n\equiv 6 \pmod p$ and $n\not \equiv 1,-3 \pmod p.$ Then for $a^{q+1}=-1,\delta=1,\gamma^{q+1}=1-n$, $\theta=1,\beta=2,\alpha=\frac{-3-2n}{\gamma^q}, \lambda=0,$ the matrix $G_n$ generates a $q^2-$ary Hermitian self-dual $[2n+2,n+1]$ code.
\end{cor}
\begin{proof} The result follows by plugging $\theta=1$ in Theorem \ref{thm:double4} 1).
\end{proof}
\begin{cor}
Assume that $n\equiv 3 \pmod p$ and $n\not \equiv 0,-3 \pmod p.$ Then for $a^{q+1}=-1,\delta=1,\gamma^{q+1}=1-n$, $ \theta=2,\beta=1,\alpha=\frac{-3-n}{\gamma^q},\lambda=0,$ the matrix $G_n$ generates a $q^2-$ary Hermitian self-dual $[2n+2,n+1]$ code.
\end{cor}
\begin{proof} The result follows by plugging $\theta=2$ in Theorem \ref{thm:double4} 2).
\end{proof}
\begin{rem} When the characteristic of the field is even, Construction (\ref{eq:double1}) and Construction (\ref{eq:double2}) coincide. Moreover the construction exists only for $n$ being odd.
\end{rem}

\section{Hermitian self-dual matrix product codes}\label{section:matrix-product}
For linear codes $(C_i)_{1\le i \le l}$, the matrix-product (MP) code $C = [C_1,\hdots,C_l]A$ is defined as a linear code whose all codewords are matrix product $[c_1,\hdots,c_l]A,$ where $c_i \in C_i$ is an $n\times 1$ column vector and $A = (a_{ij})_{l\times m}$ is an $l\times m$ matrix in $M_{m \times l} (\F_q).$ Here $l\leq  m$ and $C_i$ is an $[n,k_i ]$ code. If $C_1,\hdots,C_l$ are linear with generator matrices $G_1,\hdots,G_l,$ respectively, then $[C_1,\hdots,C_l]A$ is linear with generator matrix
$$
G=\left(
\begin{array}{cccc}
a_{11}G_1 &a_{12}G_1&\cdots&a_{1m}G_1\\
a_{21}G_2 &a_{22}G_2&\cdots&a_{2m}G_2\\
\vdots&\vdots&\cdots&\vdots\\
a_{l1}G_l &a_{l2}G_l&\cdots&a_{lm}G_l\\

\end{array}
\right).
$$

%

The lower bound on minimum distance of a matrix product code is given as follows.
\begin{thm}[\cite{FanLingLiu}]
\label{thm:lowerbound-FanLingLiu} Let $C_i$ be an $[n,k_i]$ code for $i = 1, \hdots,l,$ and let $A = (a_{i j} )_{l \times m}$
be an full row rank (FRR) matrix. Then $C = [C_1,\hdots ,C_l]A$ is an
$[nl,k_1+\cdots+k_l]$ code with minimum distance $d(C)$ satisfying 
$$d(C)\ge \min\{d(C_i)d(U_A(i))| 1\le i \le l\},$$
$$d(C)\ge \min\{d(C_i)d(L_A(i))| 1\le i \le l\},$$
where $U_A(i)$ (resp. $L_A(i)$) is a subcode generated by $A_1,\hdots,A_i$ (resp. $A_i,\hdots,A_l$) with $(A_j)_{1\le j\le l}$ ) being the $j-$th row of $A.$
\end{thm}

The Hermitian dual of a matrix product code can be determined as follow.
\begin{lem}\label{lem:hermitian-dual-MP}
Let $(C_i)_{1\le i \le l}$ be linear codes  of length $n$, and let $A\in M_{l\times m} (\F_{q^2})$ be FRR.
Assume that $B\in M_{m \times l} (\F_{q^2})$ is a right conjugate inverse of $A$, that is $A \overline{B}=I_l,$ and $H\in M_{m-l\times m} (\F_{q^2})$ is a generator
matrix of the Hermitian dual code $ L_A(1)^{\perp_H}$ of $L_A(1).$ Then the Hermitian dual of
$C = [C_1, \hdots,C_l]A$ is $$C^{\perp_H}=\left[C_1^{\perp_H},\hdots,C_l^{\perp_H},\underbrace{{\F_{q^2}^n,\hdots,\F_{q^2}^n}}_{m-l}\right]\left(\begin{array}{c}
{B}^\top\\H \end{array}\right).$$
\end{lem}

\begin{proof} With  $B$ and $H$ given, put  $B''=\left(\begin{array}{cc}
 {B}&H^\top \end{array}\right).$ Choose $A'\in M_{{(m-l)\times m}}$ such that $A' \overline{B''}=\left(\begin{array}{cc}
 A'\overline{B}&A'\overline{H^\top} \end{array}\right)=\left(\begin{array}{cc}
0&I_{m-l} \end{array}\right).$ 
Put $A''=\left(\begin{array}{c}
{A}\\A' \end{array}\right).$  
Since $A\overline{B''} =\left(\begin{array}{cc}
 A\overline{B}&A\overline{H^\top} \end{array}\right)=\left(\begin{array}{cc}
I_l&0 \end{array}\right)$, we get 
\begin{equation}
A''\overline{B''}=\left(\begin{array}{c}
{A}\\A' \end{array}\right)
\left(\begin{array}{cc}
 \overline{B}&\overline{H^\top} \end{array}\right)
=I_m.
\end{equation}

First note that $C = [C_1,\hdots ,C_l]A$ can also be written as 

$$C = [C_1,\hdots ,C_l,\underbrace{0,\hdots,0}_{m-l}]\left(\begin{array}{c}
{A}\\A' \end{array}\right)$$
We want to show that $$C^{\perp_H}=\left[C_1^{\perp_H},\hdots,C_l^{\perp_H},\underbrace{{\F_{q^2}^n,\hdots,\F_{q^2}^n}}_{m-l}\right]{B''}^\top .$$
Let ${\bf c}=(c_1\hdots,c_l,0\hdots,0)A''\in C$ and ${\bf x}=(x_1,\hdots,x_l,x_{l+1},\hdots,x_m){B''}^\top,$ where $c_i \in C_i $ for $1\le i \le l$, $x_i \in C_i^{\perp_H} $ for $1\le i \le l$ and $x_i \in \F_{q^2}^n $ for $l+1\le i \le m.$  Then we have
$$
\begin{array}{ll}
{\bf c*x}&=(c_1\hdots,c_l,0\hdots,0)A'' \left((x_1^q,\hdots,x_l^q,x_{l+1}^q,\hdots,x_m^q,)\overline{B''}^\top\right)^\top\\
&= (c_1\hdots,c_l,0\hdots,0)A''\overline{B''} (x_1^q,\hdots,x_l^q,x_{l+1}^q,\hdots,x_m^q)^\top\\
&=(c_1\hdots,c_l) (x_1^q,\hdots,x_l^q)^\top~(\text{since }A''\overline{B''}=I_m)\\
&= c_1*x_1+\cdots+ c_l*x_l=0~(\text{since } c_i\in C_i,x_i\in C_i^{\perp_H}).\\
\end{array}
$$
Thus we have shown that 

\begin{equation}
\left[C_1^{\perp_H},\hdots,C_l^{\perp_H},\underbrace{{\F_{q^2}^n,\hdots,\F_{q^2}^n}}_{m-l}\right]\left(\begin{array}{c}
 {B}^\top\\H \end{array}\right)\subset C^{\perp_H}.
\label{eq:dual-inclusion}
\end{equation}
The size of the left part of Eq. (\ref{eq:dual-inclusion}) is $q^{2(n-k_1)}\cdots q^{2(n-k_l)}q^{2n(m-l)}=q^{2\left((nm)-(k_1+\cdots+k_l)\right)}.$ Since $C$ has parameters $[nm,k_1+\cdots+k_l]$, the size of the right part of Eq. (\ref{eq:dual-inclusion}) is $q^{2\left(nm-(k_1+\cdots+k_l)\right)}.$ Thus equality holds for Eq. (\ref{eq:dual-inclusion}) and this completes the proof.
\end{proof}

 Hermitian self-dual MP codes can be characterized as follows.
\begin{thm}\label{thm:matrixproduct} Let $C_1,C_2,\hdots ,C_m$ be linear codes of length $n,$ $A\in {\cal U}_m(q^2)$ and  
$$
C=[C_1,C_2,\hdots,C_m]A.
$$
Then $C$ is a $q^2-$ary Hermitian self-dual MP code if $C_1,C_2,\hdots,C_l$ are all $q^2-$ary Hermitian self-dual codes.
\end{thm}

\begin{proof}
 We have $A\overline{A^\top}=I_m$ and from Lemma \ref{lem:hermitian-dual-MP}, we get $C^{\perp_H}=[C_1^{\perp_H},\hdots,C_m^{\perp_H}] A.$ Since $C_i$ is a Hermitian self-dual code for $1\le i\le m,$ the equality  $C_i= C_i^{\perp_H}\text{ holds for }1\le i \le m.$ Thus $C=[C_1,\hdots,C_m]A= [C_1^{\perp_H},\hdots,C_m^{\perp_H}] A=C^{\perp_H}$ and the result follows.
\end{proof}

\begin{thm}Let $C_1,C_2,\hdots ,C_m$ be linear codes of length $n, A\in M_{m\times m}(\F_{q^2})$  such that $A\overline{A}^\top =\text{diag}(a_1^q\hdots,a_m^q)$ with $a_i\in \F_{q^2}^*, {1\le i\le m}$ and  
$$
C=[C_1,C_2,\hdots,C_m]A.
$$
 Then $C$ is a Hermitian self-dual MP code if $C_1,C_2,\hdots,C_l$ are all Hermitian self-dual codes.
\end{thm}

\begin{proof} Assume that $A\overline{A}^\top =\text{diag}(a_1^q\hdots,a_m^q)$ with $a_i\in \F_{q^2}^*, {1\le i\le m}.$ Then we get 
$$
\begin{array}{ll}
A\overline{\left(A^\top\right)} D'
&=I_m,\text{ where }D'=\text{diag}(a_1^{-q},\hdots, a_m^{-q}).\\
A\overline{\left(D''A^\top\right)}
&=I_m,\text{ where }D''=\text{diag}(a_1^{-1},\hdots, a_m^{-1}).\\
\end{array}
$$ 

From Lemma \ref{lem:hermitian-dual-MP}, we have
$$
\begin{array}{ll}
C^{\perp_H}&=[C_1^{\perp_H},\hdots,C_m^{\perp_H}]\left(AD''\right)\\
&=[C_1^{\perp_H},\hdots,C_m^{\perp_H}]\text{diag}(a_1^{-1},\hdots,a_m^{-1})A\\
&=[a_1^{-1}C_1^{\perp_H},\hdots,a_m^{-1}C_m^{\perp_H}]A.\\
\end{array}
$$
Since $C_i^{\perp_H}$ is linear, we have $a_i^{-1}C_1^{\perp_H}=C_i^{\perp_H}$ for $1\le i \le m$ and it implies that $$C^{\perp_H}=[C_1^{\perp_H},\hdots,C_m^{\perp_H}]A.$$ The rest follows from the same reasoning as in Theorem \ref{thm:matrixproduct}.
\end{proof}

\begin{cor}\label{thm:matrixproduct-2}  Let $C_1,C_2,\hdots ,C_l$ be linear codes of length $n, A\in {\cal U}_m(q^2),A^{(l)}$ a submatrix of $A$ of order $l\times m$ and  
\begin{equation}
C=[C_1,C_2,\hdots,C_l]A^{(l)}.
\end{equation} Then $C$ is a Hermitian self-orthogonal MP code if $C_1,C_2,\hdots,C_l$ are all Hermitian self-orthogonal codes.
\end{cor}
\begin{proof} Let $H$ be a generator matrix of the dual code of $L_{A^{(l)}}(1)$. First note that $C$ can also be written as 
$$C=[C_1,\hdots,C_l,\underbrace{0\hdots,0}_{m-l}]\left(\begin{array}{c}A^{(l)}\\H \end{array}\right).$$ 
 From Lemma \ref{lem:hermitian-dual-MP}, $C^{\perp_H}=[C_1^{\perp_H},\hdots,C_l^{\perp_H},\underbrace{\F_{q^2}^n,\hdots,\F_{q^2}^n}_{m-l}] \left(\begin{array}{c}A^{(l)}\\H \end{array}\right).$ Since $C_i\subset C_i^{\perp_H}$ for $1\le i \le l$, we get $C\subset C^{\perp_H},$ and $C$ is a Hermitan self-orthogonal MP code.
\end{proof}

Similarly we have the following characterization.
\begin{cor} \label{cor:Hermitian self-dual MP code2}Let $C_1,C_2,\hdots ,C_l$ be linear codes. Let $A\in M_{l\times m}(\F_{q^2})$ such that $A\overline{A}^\top=\text{diag}(a_1^q,\hdots,a_l^q)$ with $a_i\not= 0$ for all $1\le i \le l$. Then $C=[C_1,C_2,\hdots,C_l]A$ is a Hermitian self-dual MP code if $C_1,C_2,\hdots,C_l$ are all Hermitian self-dual codes.
\end{cor}
\begin{proof} Let $H$ be a generator matrix of the dual code of $L_A(1)$. Assume that $A\in M_{l\times m}(\F_{q^2})$ satisfying $A\overline{A}^\top=\text{diag}(a_1^q,\hdots,a_l^q)$ with $a_i\not= 0$ for all $1\le i \le l$. Then from Lemma \ref{lem:hermitian-dual-MP}, we get 
$$C^{\perp_H}=[C_1^{\perp_H},\hdots,C_l^{\perp_H},\underbrace{\F_{q^2}^n,\hdots,\F_{q^2}^n}_{m-l}]\left(\begin{array}{c}AD''\\H \end{array}\right),$$
where $D''=\text{diag}(a_1^{-1},\hdots, a_m^{-1}).$
The rest follows with the same reasoning as that in Theorem \ref{thm:matrixproduct-2} .
\end{proof}


\section{Embedding Hermitian self-orthogonal codes}\label{section:embedding}

\begin{lem}[Witt\cite{Serre}]\label{thm:witt}
 Let $V$ be a finite dimensional vector space over $K$, where $K$ is a field of characteristic not equal to $2$. Let $s$ be a non-degenerate symmetric bilinear form on $V$ and let  $W \subset V$ be a subspace of $V.$ Assume that $\phi : W\longrightarrow V$ is an isometry, that is for any ${\bf x,y} \in W, s(\phi({\bf x},\phi({\bf y}))=s(x,y)$. Then $\phi$ can be extended to an isometry $\phi' : V\longrightarrow V.$
\end{lem}

Pless showed that under certain conditions, an analog of Witt's Theorem holds in characteristic 2. In the particular case, where $K$ a finite field of characteristic 2, it gives the following:
\begin{lem}[Pless\cite{Pless}]  Let $V$ be a a finite dimensional vector space over $K$, where $K$ is a field of characteristic equal to $2$. Let $s$ be a non-degenerate symmetric bilinear form on $V$ and let  $W \subset V$ be a subspace of $V.$ Assume that $\phi : W\longrightarrow V$ is an isometry. Assume moreover that the following holds: If $1=(1,\hdots,1) \in W,$ then $\phi(1)=1$ Otherwise, if $1\not\in W,$ then $1\not\in \phi(W).$  Then $\phi$ can be extended to an isometry $\phi' : V\longrightarrow V.$
\end{lem}
As a consequence of the above lemmas, we have the following result.
\begin{prop}\label{prop:embedding} Let $C$ be a $q^2-$ary Hermitian self-orthogonal $[n,k]$ code with $n$ being even and $k\le \lfloor\frac{n-1}{2}\rfloor.$ Then there exists a $q^2-$ary Hermitian self-dual code of length $n.$
\end{prop}

\begin{proof} Let $\alpha=\omega ^{\frac{q-1}{2}},$ where $\omega$ is a primitive element of $\F_{q^2}.$ Let $D$ be a Hermitian self-dual $[n,\frac{n}{2}]$ code 
with its generator matrix as follows:
$$
\left(
\begin{array}{ccccccc}
1&\alpha&0&0&\cdots&0&0\\
0&0&1&\alpha&\cdots&0&0\\
\vdots&\vdots&\vdots&\vdots&\vdots&\vdots&\vdots\\
0&0&0&0&\cdots&1&\alpha\\
\end{array}
\right)
$$
 As $\dim C \le \dim D,$ there is an injective linear map $\phi : C\longrightarrow D.$ Then for any ${\bf x},{\bf y} \in C$,  we have $\phi({\bf x})*\phi({\bf y})={\bf x}*{\bf y}=0$ since both $C$ and $D$ are Hermitian self-orthogonal. Thus $\phi$ is an isometry and by the above two lemmas, $\phi$ can be extended as $\phi':D\longrightarrow D.$ Taking $C'=\phi'^{-1}(D)$, the result follows.

\end{proof}

We can deduce the construction of  $q^2-$ary Hermitian self-dual codes of length $n$ (if they exist) as follows.
\begin{cor}\label{cor:embedding2} Let $n$ be an odd positive integer and $C$ a Hermitian self-orthogonal $[2n-1,n-1,d]$ code. Then there exists a Hermitian self-orthogonal  $[2n,n,d]$ code $C_0$ which can be embedded into a Hermitian self-dual $[2n,n,d]$ code $C'$  such that $C_0\subset C'\subset C_0^{\perp_H}.$
\end{cor}
\begin{proof} Let $G$ be the generator matrix of $C$ and $C_0$ be a Hermitian self-orthogonal code obtained from $C$ by lengthening one zero coordinate. Clearly the code $C_0$ has parameters $[2n,n-1,d]$ and $C_0^{\perp_H}$ has parameters $[2n,n+1]$. Denote $G_0$ the generator matrix of $C_0,$ that is 
$$G_0=
\left(
\begin{array}{cc}
&0\\
G&\vdots\\
&0\\
\end{array}
\right).
$$
Let ${\bf x}\in C_0^{\perp_H}\slash C_0$ such that ${\bf x}* {\bf x}=0$. Then the code $C_0'$ with its following generator matrix $G_0'$ is self-dual with parameters $[2n, n]:$

$$G'_0=
\left(
\begin{array}{ccc}
&&0\\
G&&\vdots\\
&&0\\
\hline
&{\bf x}&\\
\end{array}
\right).
$$
Moreover the following inclusion holds:
$$C_0\subset C'_0\subset {C_0}^{\perp_H}.$$
\end{proof}
Simarly to the above corollary, we have:
\begin{cor}\label{cor:embedding3} Let $C$ be a Hermitian self-orthogonal $[2n,n-1,d]$ code. Then there exists a Hermitian self-dual $[2n,n]$ code $C'$  such that $C\subset C'\subset C'^{\perp_H}.$
\end{cor}

It is clear that all rows of an $n\times n$ matrix $L\in {\cal U}_q(n)$ span the ambient space ${\mathbb F}_q^n$. The space spanned by rows $L$ can be embedded and thus we have another construction of a Hermitian self-dual $[2n+2,n+1]$ code by coordinate extension as follows.

\begin{prop}\label{prop:recursive3}  Let $q=p^m$ with $p$ being a prime. Fix $a\in {\mathbb F}_q$ such that $a^{q+1}= -1 $.  Let $L \in {\cal U}_n(q^2)$ and $\lambda_1,\hdots, \lambda_{n}  \in {\mathbb F}_q$. Let ${\bf x}$ be a vector of length $n+2$ satisfying ${\bf x* x }=0$ and $L'_i*{\bf x}=0,\forall 1\le i\le n,$ where $L'_i=(aL_i|a\lambda_i,\lambda_i)$ is the extended row $i$ of $aL.$
Then the code with the following generator matrix is a $q^2-$ary Hermitian self-dual $[2n+2,n+1]$ code:
\begin{equation} \label{eq:recursive3}
G'_{n}=\left(\begin{array}{ccc|ccccc}
&&&&&&\lambda_1a&\lambda_1\\
&I_n&&&aL&&\vdots&\vdots \\
&&&&&&\lambda_{n}a&\lambda_{n}\\
\hline
&{\bf 0}&&&&{\bf x} &&\\
\end{array}\right).
\end{equation}
\end{prop}
\begin{proof}
For $q=p^m$, such a vector $\bf x$ of length $n+2$ exists, for example ${\bf x}=(0,\hdots,0,a,1)$ satisfies the desired condition.

The rank of $G'_{n}$ is obviously $n+1$ and each row of $G'_{n}$ is Hermitian orthogonal to itself and to other rows and thus the result follows.
\end{proof}
\begin{rem}
\begin{enumerate}
\item The vector ${\bf x}_0=(0,\hdots,0,a,1)$ satisfies the above properties but is not interesting for constructing optimal codes since the code constructed from this vector will have minimum distance at most $2$. However it is useful for generating many vectors with Hamming weight greater than $2$. 
\item Those vectors $\bf x$ can be found algorithmically for example echolonizing the matrix obtained from all  the extended rows $L'_i=(aL_i|a\lambda_i,\lambda_i)$ of $aL$ allows us to calculate a vector ${\bf y}=(y_1,\hdots,y_n)$ such that $L'_i*{\bf y}=0$ for $1\le i\le n.$ 
For  all $\alpha_0,\alpha_1\in \F_{q^2}$, we can search for vectors of the form ${\bf x}=\alpha_0 {\bf x}_0+\alpha_1{\bf y}$  satisfying ${\bf x}*{\bf x}=0$.
\end{enumerate}
\end{rem}

By taking $\lambda_1=\cdots=\lambda_n=1$, we get an explicit construction of such a vector $\bf x$ in Construction (\ref{eq:recursive3}) and thus the following construction.

\begin{prop}  Let $q=p^m$ with $p$ being a prime. Fix $a\in {\mathbb F}_q$ such that $a^{q+1}= -1$.  Let $L \in {\cal U}_n(q^2)$ 
Then for any $n\equiv 1 \pmod p$, the code with the following generator matrix is a $q^2-$ary  Hermitian self-dual $[2n+2,n+1]$ code:
\begin{equation} \label{eq:recursive34}
G'_{n}=\left(\begin{array}{ccc|ccccc}
&&&&&&a&1\\
&I_n&&&aL&&\vdots&\vdots \\
&&&&&&a&1\\
\hline
&{\bf 0}&&&&{\bf x} &&\\
\end{array}\right),
\end{equation}
where ${\bf x}$ is a vector of length $n+2$ belonging to $X=\{ {\bf z}=\alpha(aL_1+\cdots+aL_n,0,1)+\beta(0,\hdots,0,a,1)|  \alpha,\beta \in \F_{q},{\bf z}* {\bf z}=0\}$ and $L_i$ is the $i-$th row of $L.$ In particular (with ${\bf x}=(aL_1+\cdots+aL_n,0,1)$)
the code with the following generator matrix is a $q^2-$ary Hermitian self-dual $[2n+2,n+1]$ code:

\begin{equation} \label{eq:recursive33}
G'_{n}=\left(\begin{array}{ccc|c|cc}
&&&&a&1\\
&I_n&&aL&\vdots&\vdots \\
&&&&a&1\\

\hline
&{\bf 0}&&a(L_1+\cdots+L_n) &0&1\\
\end{array}\right),
\end{equation}
\end{prop}

\begin{proof} Take $g_i$ as the $i-$th row of $G'_{n}$.  Take  ${\bf x}_1=(0,\hdots,0,aL_1+\cdots+aL_n,0,1),{\bf x}_0=(0,\hdots,0,a,1)\in \F_{q^2}^{2n+2}.$ Clearly ${\bf x}_0*{\bf x}_0={\bf x}_1*{\bf x}_1=0.$ The $(n+1)-$th row of $G'_n$ can be written as $g_{n+1}=\alpha{\bf x}_0+\beta{\bf x}_1$ for some $\alpha,\beta\in \F_{q^2}$ and thus $g_i*g_{n+1}=0$ for $1\le i\le n$ and $g_{n+1}*g_{n+1}=0.$ Hence the result follows.
\end{proof}

The following algorithm is used to embed a Hermitian self-orthogonal $[2n+2,n]$ code into a Hermitian self-dual $[2n+2,n+1,d]$ code.

\noindent \rule[-.1cm]{\linewidth}{0.2mm}
{\bf Algorithm 2}\\
\noindent \rule[.3cm]{\linewidth}{0.2mm}

\vspace{-.3cm}
\noindent\textbf{Input:} $L$: an $n\times n$ unitary matrix, $\lambda=(\lambda_1,\hdots,\lambda_n):$ a vector of length $n$\\
\textbf{Output:} A list $X$ of vectors of length $n+2$

\vspace{.2cm}
\noindent
\begin{small}
\begin{enumerate}
\item ${\bf x}_0:=(0,\hdots,0,a,1)$;
\item Extending $L$ to $L'$ as in Proposition \ref{prop:recursive3};
\item M:=EchelonForm($L'$);
\item Determining a vector ${\bf y}\not=0$ satisfying ${\bf y}*M_i=0$ for all $1\le i\le n$;
\item $X:=[~]$;
\item  for $\alpha$ in $ \F_{q^2}$ do 
\item for $\beta$ in $\F_{q^2}$ do 
\item ${\bf x}:=\alpha {\bf x}_0+\beta {\bf y}$;
\item if ${\bf x}*{\bf x}=0$ then 
\item {\tt Append}$(X,{\bf x})$;
\item end if;
\item end if;
\item end for;
\item end for;
\item return X;
\end{enumerate}
\end{small}
\noindent \rule[.3cm]{\linewidth}{0.2mm}
We now prove that the algorithm correctly returns a list of vectors ${\bf x}=\alpha {\bf x}_0+\beta {\bf y}$ for some $\alpha,\beta\in \F_{q^2}$. It is enough to show that such a non-zero vector {\bf y} exists. Since $L$ is unitary, it is invertible and by applying elementary row operations on the extended matrix $L'$, we get, at Step 3 of the algorithm,

\begin{equation}
M=\left(
\begin{array}{cccccc}
1&0&\cdots&0&m_{1,n+1}&m_{1,n+2}\\
0&1&\cdots&0&m_{2,n+1}&m_{2,n+2}\\
\vdots&\vdots&\ddots&\vdots&\vdots&\vdots\\
0&0&\cdots&1&m_{n,n+1}&m_{n,n+2}\\
\end{array}
\right)\label{eq:echelon}
\end{equation}
At Step 4, we determine the coordinates of ${\bf y}=(y_1,\hdots,y_n,y_{n+1},y_{n+2})$ from Eq. (\ref{eq:echelon}) as follows: 

\begin{itemize}
\item $y_1=1,y_{n+1}=0,y_{n+2}=-\frac{1}{m_{1,n+2}^q}.$
\item $y_2=-y_{n+2}m_{2,n+2}^q.$\\
 $\vdots$
\item $y_n=-y_{n+2}m_{n,n+2}^q.$
\end{itemize}
It can be easily checked that $(y_1,\hdots,y_n,y_{n+1},y_{n+2})*M_i=0$ for all $1\le i \le n.$
Note that for $k<n$, the proof of ${\bf y}*M_i=0$ for all $1\le i \le k,$ follows with the same approach by taking $y_1=1,y_{k+1}=\cdots=y_{n+1}=0,y_{n+2}=-\frac{1}{m_{1,n+2}^q}.$

\section{Numerics}\label{section:num}
All our constructions are based on the unitary matrices in the group. For small dimensions, say $n\le 3$ or $n=4$ and $q=2^2,3^2$, such matrices can be easily reachable and thus optimal Hermitian self-dual codes can be constructed efficiently. For large dimension $n$, we can no longer visit all the elements in the group. However by applying Algorithm 1, we can efficiently construct unitary matrices with good properties, say non-sparse matrices, which are really useful for constructing MDS codes when applying Construction (\ref{eq:sd2}). Recall that all the Hermitian self-dual codes obtained from our constructions use the unitary matrix $L$ computed by Algorithm 1 and such a matrix is of the form $L=N^iP^jQ^kR^l$.\\

\begin{itemize}
\item Parameters of Hermitian self-dual codes over $\F_{2^2}$ and $\F_{5^2}$ from Construction (\ref{eq:sd1}) are given in Table \ref{table:F4} and Table \ref{table:F25} respectively and are compared with the quadratic double circulant construction \cite{Gabo} and the building-up construction \cite{KimLee} respectively. With the same parameters available, our constructions perform better than those of \cite{Gabo} and \cite{KimLee}. More precisely over $\F_{2^2}$ we obtain a Hermitian self-dual code with parameters $[18,9,8]$ while the parameters in \cite{Gabo} are just $[16,8,6].$ Also over $\F_{3^2}$ we obtain a Hermitian self-dual code with parameters $[10,5,6]$ but the parameters in \cite{KimLee} are just $[10,5,5].$
\item Parameters of Hermitian self-dual codes over $\F_{4^2}$ from Construction (\ref{eq:sd2}) are given in Table \ref{table:F16} and are compared with the building-up construction \cite{KimLee}. We obtain an MDS Hermitian self-dual code with parameters $[10,5,6]$ which are better than $[10,5,5]$ in \cite{KimLee}.
\item  In Table \ref{table:F49-81}, we provide parameters of Hermitian self-dual codes over $\F_{q^2},q^2=5^2,7^2,8^2,9^2,11^2,13^2$ up to length $18$ from Construction (\ref{eq:recursive3}) and compare with \cite{GulKimLee}. We obtain MDS Hermitian self-dual codes over $\F_{8^2},\F_{11^2},\F_{13^2},$ with parameters $[12,6,7],[14,7,8],[14,7,8]$ respectively. These parameters are better than \cite{GulKimLee} and meanwhile the first two parameters are new in \cite{TongWang}. We also find MDS Hermitian self-dual codes over $\F_{16^2},\F_{17^2}$ with parameters $[14,7,8]$ which are better than \cite{GulKimLee}. We give examples of the two MDS codes with new parameters as follows.

\begin{exam}  For $q=8^2, \omega$ being a primitive element of $\F_{8^2}$ and  $n=5,L=N^0P^ 1 Q^1 R^1,a=1,(\lambda_1,\hdots,\lambda_n)=(\omega^{21}, \omega^{44}, 1, 1, 1)$, we find $ { \bf x}=(1, \omega^{19}, \omega^{37}, \omega^{62}, \omega^{57}, \omega^{37}, \omega^{42 })$ and a Hermitian self-dual code with parameters $[12,6,7]$. The generator matrix of such a code is given by:
$$
\left(
\begin{array}{ccc|ccccccc}

 &&&&\omega^{27}& \omega^{55}& \omega^{16}& \omega^{32}& \omega^{18}& \omega^{55 }\\ 
    
&&&&  \omega^{17}& \omega^{22}& \omega^{58}& \omega^{28}& \omega^{43}& \omega^{54 }\\ 
    
&I_6&&& \omega^{41}& \omega^{58}& \omega^{6}& \omega^{43}& \omega^{34}& \omega^{57 }\\ 
    
&&&& \omega^{14}& \omega^{5}& \omega^{59}& \omega^{30}& \omega^{11}& \omega^{32 }\\ 
    
 &&&&\omega^{9}& \omega^{35}& \omega^{29}& \omega& \omega^{36}& \omega^{53 }\\
&&&&  \omega^{19}& \omega^{37}& \omega^{62}& \omega^{57}& \omega^{37}& \omega^{42 }\\ 
\end{array}
\right).
$$
\end{exam}

\begin{exam}  For $q=11^2, \omega$ being a primitive element of $\F_{11^2}$ and  $n=6,L=N^0P^ 1 Q^2 R^4,a=\omega^{5},(\lambda_1,\hdots,\lambda_n)=(\omega^{25}, \omega^{61}, \omega^{46}, \omega^{86}, \omega^{42}, \omega^{100})$, we find $ { \bf x}=(\omega^{38}, \omega^{9}, \omega^{93}, \omega^{34}, \omega^{91}, \omega^{91}, 
\omega^{9}, \omega^{100})$ and a Hermitian self-dual code with parameters $[14,7,8]$. The generator matrix of such a code is given by:
$$
\left(
\begin{array}{ccc|cccccccc}

&&&&\omega^{85}& \omega^{86}& \omega^{52}& \omega^{65}& \omega^{97}& \omega^{76}& \omega^{54}\\
&&&&\omega^{11}& \omega^{92}& \omega^{118}& \omega^{116}& \omega^{5}& \omega^{58}& \omega^{50}\\
&&&& \omega^{94}& \omega^{94}& \omega^{110}& \omega^{39}& \omega^{53}& 1& \omega^{29}\\
&I_7&&&\omega^{68}& \omega^{15}& \omega^{119}& \omega^{47}& \omega^{109}& \omega^{25}& \omega^{30}\\
&&&&\omega^{49}& \omega^{66}& \omega^{102}& \omega^{86}& \omega^{17}& \omega^{105}& \omega^{116}\\
&&&&2 &\omega^{77}& \omega^{4}& 7 &\omega^{33}& \omega^{73}& \omega^{119}\\
&&&&\omega^{91}& \omega^{55}& \omega^{116}& \omega^{53}& \omega^{53}& \omega^{91}& \omega^{62}\\
\end{array}
\right).
$$
\end{exam}

\item Parameters of Hermitian self-dual codes over $\F_{3^2}$ for Theorem \ref{thm:double1} 1) are given in Table \ref{table:F9} and over $\F_{11^2},\F_{13^2}$ for Theorem \ref{thm:double1} 2) and 3) in Table \ref{table:new}.

\item Parameters of Hermitian self-dual codes over $\F_{2^2}$ and $\F_{4^2}$ for Theorem \ref{thm:double2} 4) are given in Table \ref{table:F4} and Table \ref{table:F16} respectively.

\item Parameters of Hermitian self-dual codes over $\F_{3^2}$ and $\F_{5^2}$ for Theorem \ref{thm:double3} 1)  and 3) are given in Tables \ref{table:F9} and \ref{table:F25} respectively.

\item Parameters of Hermitian self-dual codes over $\F_{17^2}$ and $\F_{19^2}$ for Theorem \ref{thm:double4} 4) are given in Table \ref{table:new}.

\end{itemize}

\begin{exam} We construct a Hermitian self-dual MP code over $\F_{2^2}$ as follows. Take a unitary matrix $A$ and $7$ Hermitian self-dual codes $C_1,\hdots,C_7$ of length $8$:\\
{\tiny
$
A=\left(
\begin{array}{ccccccc}
 \omega&  \omega&\omega^2&\omega&  \omega&\omega^2&\omega^2\\
  0&1&\omega^2&\omega^2&\omega&\omega^2&0\\
  1&\omega^2&0&\omega^2&\omega^2&0&\omega^2\\
  1&1&0&0&1&1&\omega^2\\
  1&\omega^2&\omega^2&\omega&  \omega&  0&0\\
  \omega&  0&\omega&  0&\omega&  1&\omega^2\\
  0&0&\omega&  \omega&\omega^2&1&1\\
\end{array}
\right),
$}
{\scriptsize
 $
C_1=\left(
\begin{array}{cccccccc} 
 &0&\omega&\omega^2&1\\
  I_4  &1&0&\omega^2&1\\
&\omega&\omega^2&1&0\\
  &\omega&\omega^2&0&\omega\\
\end{array}
\right),
$
$
C_2=\left(
\begin{array}{cccccccc} 
 &1&\omega^2&1&0\\
I_4    &0&\omega^2&1&\omega\\
  &1&0&1&\omega\\
  &\omega^2&\omega&  0&1\\
\end{array}
\right),
$
$
C_3=\left(
\begin{array}{cccccccc} 
 &0&\omega&\omega^2&\omega\\
  I_4  &1&0&1&\omega^2\\
  &\omega&  1&0&1\\
  &\omega^2&\omega&\omega^2&0\\
\end{array}
\right),
$
}
{\scriptsize
$
C_4=\left(
\begin{array}{cccccccc} 
 &1&\omega^2&1&0\\
 I_4  &0&\omega^2&1&\omega\\
   &1&0&1&\omega\\
  &\omega^2&\omega&  0&1\\
\end{array}
\right),
$
$
C_5=\left(
\begin{array}{cccccccc} 
 &0&\omega&  1&\omega\\
  I_4  &\omega^2&\omega&  0&\omega\\
  &\omega&  1&\omega^2&0\\
  &\omega&  0&\omega^2&1\\
\end{array}
\right),
$
$
C_6=\left(
\begin{array}{cccccccc} 
 &\omega&  0&\omega^2&1\\
 I_4   &\omega&  \omega&  0&1\\
  &\omega^2&\omega^2&1&0\\
  &0&\omega^2&1&\omega\\
\end{array}
\right),
$
$
C_7=\left(
\begin{array}{cccccccc} 
 &\omega&  1&0&\omega\\
  I_4 &0&\omega^2&\omega&  1\\
   &\omega^2&\omega&  1&0\\
  &\omega&  0&\omega^2&\omega\\
\end{array}
\right).
$
}

The code $[C_1,\hdots,C_7]A$ is Hermitian self-dual with parameters $[56,28,14]$ and it is an optimal code as per \cite{Gweb4H}.
\end{exam}

Parameters of the new MDS and almost MDS as well as optimal Hermitian self-dual codes from all our constructions are summarized in Table \ref{table:new00} and Table \ref{table:new01} while the overall parameters are given in the Appendix. 
\begin{rem}
\begin{enumerate}
\item As we see in the tables of the Appendix,  optimal Hermitian self-dual codes can be constructed very efficiently from less than $5^4$ unitary matrices  which is extremely small compared with the order of the unitary group considered.
\item Compared with other constructions, Construction (\ref{eq:recursive3}) may produce MDS Hemitian self-dual codes more efficiently for example with length $14$ over $\F_{11^2},\F_{13^2}$. It also happens for the almost MDS Hermitian self-dual code over $\F_{5^2}$ with parameters $[16,8,8].$ See Table \ref{table:F49-81}.
 \item It should be noted that the parameters $[20,10, 8]$ over $\F_{2^2}$ are hard to be reachable since we have to considered up to $33^4$ unitary matrices (see Table \ref{table:F4})  and may not be reachable by Constructions (\ref{eq:sd1}), (\ref{eq:sd2}) and (\ref{eq:recursive3}) since we have tried the searching with the same matrices but we could not get such parameters. This shows that the construction in Theorem \ref{thm:double2} 4) performs better than those three constructions.
\item In the double circulant construction \cite{Gabo}, the exhaustive search for optimal Hermitian self-dual codes over $\F_{q}$ takes $q^3$ matrices while most of ours take only $5^4$ matrices (see the tables in the Appendix) which are much smaller if $q\ge 3^2.$

\end{enumerate}
\end{rem}
\begin{table}[h]
\caption{MDS and almost MDS Hermitian self-dual codes, M: MDS, A: almost MDS, $^{**}:$ new parameters}
$$
\begin{array}{c|c|c|c|c|c|c|c|c}
q/2n&4&6&8&10&12&14&16&18\\
\hline
3^2&M&M&A&M&A&&&\\
4^2&M&M&A&M&A&A&&\\
5^2&M&M&M&M&A&A^{**}&A^{**}&\\
7^2&M&M&M&M&M&A&A^{**}&A^{**}\\
8^2&M&M&M&M&M^{**}&A^{**}&A^{**}&A^{**}\\
9^2&M&M&M&M&M&A^{**}&A^{**}&A^{**}\\
11^2&M&M&M&M&M&M^{**}&A^{**}&A^{**}\\
13^2&M&M&M&M&M&M&A^{**}&A^{**}\\
16^2&M&M&M&M&M&M&A&A^{**}\\
17^2&M&M&M&M&M&M&A&A\\
19^2&M&M&M&M&M&M&A&A\\
\end{array}
\label{table:new00}
$$
\end{table}

\begin{table}
\caption{ Optimal Hermitian self-dual codes, $^{**}:$ new distance}
$$
\begin{array}{c|c|c|c|c|c|c|c|c}
q/2n&14&16&18&20&22&24&26&28\\
\hline
3^2&6^{**}&7^{**}&7^{**}&8^{**}&8^{**}&9^{**}&9^{**}&10^{**}\\
4^2&7&7^{**}&8^{**}&8^{**}&9^{**}&9^{**}&10^{**}&11^{**}\\
5^2&7^{**}&8^{**}&8^{**}&9^{**}&9^{**}&10^{**}&10^{**}&11^{**}\\
\end{array}
\label{table:new01}
$$
\end{table}


\section{Conclusion}\label{section:conclusion}

In this correspondence we have developed new methods and algorithms to construct Hermitian self-dual codes over large finite fields. Some constructions are similar to those in the Euclidean case \cite{ShiSokSole} while the others are the generalizations of the quadratic double circulant construction \cite{Gabo}. Most constructions \cite{ShiSokSole} used random orthogonal matrices to construct Euclidean self-dual codes but here we have provided algorithmic constructions. While the quadratic double circulant construction \cite{Gabo} is only possible with lengths of prime power, our generalized methods allow the construction of Hermitian self-dual codes for all lengths in odd characteristic and for more lengths in even characteristic. 
The numerical results give more than forty Hermitian self-dual codes with new optimal parameters. Searching more optimal codes for longer lengths requires stronger machine implementation, like distributed computing or a lower level language.\\

\noindent
{\bf Acknowledgement:} This research work is supported by Anhui Provincial Natural Science Foundation with grant number 1908085MA04.
\section*{Appendix}
 Parameters of Hermitian self-dual codes over $\F_{q}, q=2^2,3^2,4^2,5^2,7^2,8^2,9^2,11^2,13^2,17^2,19^2$ from various constructions are given as follows.
\begin{table}[h]
\caption{Hermitian self-dual codes over $\F_{2^2}$ compared with the quadratic double circulant construction \cite{Gabo}, where $L=N^iP^jQ^kR^l$ is computed by Algorithm 1 with $m=n$, $^*$: optimal distance as per \cite{Gweb4H}, $-$: not available in \cite{Gabo}}

$$
\begin{array}{cccccc|c}
n&\text{length}&\text{Distance}&\text{Distance}\cite{Gabo}&[i,j,k,l]&[\delta,\theta,\beta,\alpha,\gamma,a,\lambda]&\text{Construction }\\
\hline
\\

2&4&3^*&-&&\\

3&6&4^*&-&&&\\

4&8&4^*&-&[ 0, 0, 1, 1 ]&&\\

5&10&4^*&-&[ 0, 0, 0, 1 ]&&\\

6&12&4^*&-&[ 0, 0, 1, 1 ]&&\\

7&14&6^*&-&[ 1, 0, 2, 2 ]&&\\

8&16&6^*&-&[ 1, 0, 1, 1 ]&&\\

9&18&8^*&6&[ 9, 9, 10, 2 ]&&\\


11&22&8^*&-&[ 9, 0, 2, 2 ]&&(\ref{eq:sd1})\\

12&24&8^*&-&[ 1, 0, 2, 2 ]&&\\

13&26&8^*&-&[ 0, 0, 1, 1 ]&&\\

14&28&8&-&[ 1, 0, 1, 0 ]&&\\

15&30&8&-&[ 0, 0, 1, 1 ]&&\\


17&34&10^*&10&[ 1, 0, 2, 3 ]&&\\

\\
\hline
\\
9&20&8^*&8&[32,0,7,32]&[1,\omega,\omega,\omega,0,1,\omega]&\\
15&32&10^*&-&[8,20,20,15]&[1,\omega,\omega,\omega,0,1,\omega]&\\
21&44&12^*&-&[1,2,4,4]&[1,\omega,\omega,\omega,0,1,\omega]&(\ref{eq:double1})\text{ Theorem } \ref{thm:double2} ~4)\\
27&56&14^*&-&[1,2,4,1]&[1,\omega,\omega,\omega,0,1,\omega]&\\
33&68&16^*&-&[0,1,4,4]&[1,\omega,\omega,\omega,0,1,\omega]&\\

\end{array}
$$\label{table:F4}
\end{table}

\begin{table}
\caption{Hermitian self-dual codes over $\F_{3^2}$ compared with \cite{KimLee},  where $L=N^iP^jQ^kR^l$ is computed by Algorithm 1 with $m=n$, $-$: not available in \cite{KimLee} }
$$
\begin{array}{cccccc|c}
n&\text{length}&\text{Distance}&\text{Distance}\cite{KimLee}&[i,j,k,l]&[\delta,\theta,\beta,\alpha,\gamma,a,\lambda]&\text{Construction }\\
\hline
\\

2&4&3&3&&&\\
3&6&4&4&&& (\ref{eq:sd1})\\
5&10&6&5&[1,2,1,3]&&\\
\\
\hline
\\
3&8&4&4&[  0, 0, 0, 0 ]&[  0, \omega, \omega^3, 1, \omega, \omega, 0 ]&\\

6&14&6&6&[  1, 0, 1, 0 ]&[  0, \omega, \omega^3, 1, \omega, \omega, 0 ]&\\

7&16&7&-&[  4, 0, 4, 2 ]&[  0, \omega, \omega^2, 0, 1, \omega, 0 ]&\\

9&20&8&-&[  3, 1, 3, 1 ]&[  0, \omega, \omega^3, 1, \omega, \omega, 0 ]&\\

10&22&8&-&[  0, 0, 1, 1 ]&[  0, \omega, \omega^2, 0, 1, \omega, 0 ]&(\ref{eq:double1})\text{ Theorem } \ref{thm:double1} ~1)\\

12&26&9&-&[  0, 0, 1, 1 ]&[  0, \omega, \omega^3, 1, \omega, \omega, 0 ]&\\


13&28&10&-&[  1, 2, 3, 1 ]&[  0, \omega, \omega^2, 0, 1, \omega, 0 ]&\\

15&32&11&-&[  2, 0, 2, 3 ]&[  0, \omega, \omega^3, 1, \omega, \omega, 0 ]&\\

16&34&11&-&[  0, 0, 1, 1 ]&[  0, \omega, \omega^2, 0, 1, \omega, 0 ]&\\

\\
\hline
\\

5&12&6&6&[ 0, 2, 0, 2 ]&[  0, \omega, \omega^3, 0, \omega, \omega, 0 ]]&\\

8&18&7&-&[ 0, 1, 1, 1 ]&[  0, \omega, \omega^3, 0, \omega, \omega, 0 ]&\\

11&24&9&-&[ 0, 3, 0, 2 ]&[  0, \omega, \omega^3, 0, \omega, \omega, 0 ]&(\ref{eq:double2})\text{ Theorem } \ref{thm:double3}~ 1)\\

14&30&10&-&[ 0, 0, 1, 3 ]&[  0, \omega, \omega^3, 0, \omega, \omega, 0 ]&\\

17&36&12&-&[ 1,2,0, 3 ]&[  0, \omega, \omega^3, 0, \omega, \omega, 0 ]&\\

\end{array}
$$\label{table:F9}
\end{table}

\begin{table}
\caption{Hermitian self-dual codes over $\F_{4^2}$ compared with \cite{KimLee}, where $L=N^iP^jQ^kR^l$ is computed by Algorithm 1 with $m=n$, $-$: not available in \cite{KimLee} }
$$
\begin{array}{cccccc|c}
n&\text{length}&\text{Distance}&\text{Distance}\cite{KimLee}&[i,j,k,l]&[\delta,\theta,\beta,\alpha,\gamma,a,\lambda]&\text{Construction }\\
\hline
\\

2&4&3&3&&&\\

3&6&4&4&&&\\

4&8&4&4&[0,0,0,1 ]&&\\

5&10&6&5&[ 0, 5, 3, 0 ]&\\


7&14&7&-&[ 1, 2, 0, 0 ]&&(\ref{eq:sd2})\\


9&18&8&-&[1,0,3,1]&\\


11&22&9&-&[0,3,2,0]&\\

13&26&10&-&[ 0, 2, 2, 3 ]&\\

\\
\hline
\\

5&12&6&6&[  1, 0, 0, 1 ]&[  1, \omega, \omega, \omega, 0, 1, \omega ]&\\
7&16&7&-&[  0, 1, 1, 1 ]&[  1, \omega, \omega, \omega, 0, 1, \omega ]&\\
9&20&8&-&[  0, 1, 1, 2 ]&[  1, \omega, \omega, \omega, 0, 1, \omega ]&(\ref{eq:double1})\text{ Theorem }\ref{thm:double2}~4)\\
11&24&9&-&[  0, 1, 1, 1 ]&[  1, \omega, \omega, \omega, 0, 1, \omega ]&\\
13&28&11&-&[  3, 4, 1, 5 ]&[  1, \omega, \omega, \omega, 0, 1, \omega ]&\\

\end{array}
$$\label{table:F16}
\end{table}

\begin{table}
\caption{Hermitian self-dual codes over $\F_{5^2}$ compared with \cite{KimLee},  where $L=N^iP^jQ^kR^l$ is computed by Algorithm 1 with $m=n$, $-$: not available in \cite{KimLee} }
$$
\begin{array}{cccccc|c}
n&\text{length}&\text{Distance}&\text{Distance}\cite{KimLee}&[i,j,k,l]&[\delta,\theta,\beta,\alpha,\gamma,a,\lambda]&\text{Construction }\\

\hline\\

%

2&4&3&3&&\\

3&6&4&4&&\\

4&8&5&5&[  0, 1, 1, 0 ]&\\

5&10&6&6&[  2, 2, 2, 1 ]&\\

6&12&6&6&[  0, 1, 0, 1 ]&&(\ref{eq:sd1})\\


9&18&8&-&[  0, 1, 0, 1 ]&\\


11&22&9&-&[  0, 0, 1, 1 ]&\\



14&28&11&-&[  0, 1, 0, 1 ]&\\
\\
\hline
\\

6&14&7&-&[  0, 1, 2, 2 ]
&
[  0, \omega, \omega^{20}, \omega^3, \omega, \omega^2, \omega^{22} ]\\

9&20&9&-&[  1, 4, 3, 3 ]
&
[  0, \omega, \omega^{13}, \omega^3, \omega^2, \omega^2, \omega^{15} ]&(\ref{eq:double2})\text{ Theorem } \ref{thm:double3}~ 3)\\


11&24&10&-&[  0, 0, 3, 1 ]
&
[  0, \omega, \omega^{20}, \omega^3, \omega, \omega^2, \omega^{22} ]\\
12&26&10&-&[  0, 0, 1, 1 ]
&
[  0, \omega, \omega^9, \omega^3, 1, \omega^2, \omega^{11} ]\\



\\
\hline
\end{array}
$$\label{table:F25}
\end{table}

\begin{table}
\caption{Hermitian self-dual codes over $\F_{q}, q=5^2,7^2,8^2,9^2,11^2,13^2$ compared with \cite{KimLee},  where $L=N^iP^jQ^kR^l$ is computed by Algorithm 1 with $m=n$, $(d):$ distance in \cite{KimLee},$-$: not available in \cite{KimLee}}
$$
\begin{array}{l||llllllll|c}
q&n&\text{length}&\text{Distance}&\text{Distance}\cite{KimLee}&[i,j,k,l]&a&(\lambda_1,\hdots,\lambda_n)&\bf{x}&\text{Construction }\\

\hline\\

5^2&7&16&8&-& [ 0, 0, 1, 3] & \omega^{2} & 2 \omega^{3} 2 \omega^{19} 2 \omega^{7} \omega^{19}  
& \omega^{10} \omega^{13} 2 \omega^{4} \omega^{4} 4 \omega^{4} \omega^{3} \omega^{23} &(\ref{eq:recursive3})\\
\\
\hline
\\

&4&10&6&6&  [0, 0, 2, 1] &  \omega^{3} &   \omega^{22} \omega \omega^{46} \omega^{4}  &   \omega^{47} \omega^{39} \omega^{11} \omega^{31} \omega^{13} \omega^{27}  
&\\

&5&12&7&7&  [0, 0, 0, 1 ]&  \omega^{3} &  \omega^{44} \omega^{2} \omega^{5} \omega^{11} \omega^{28}  &   \omega^{14} \omega^{27} \omega^{18} 5 \omega^{21} \omega^{36} \omega^{18} 
&\\

7^2&6&14&7&-&  [0, 0, 1, 1] &  \omega^{3 }&   \omega^{21} \omega^{19} \omega^{17} \omega^{12} \omega^{36} \omega^{43}  &   \omega^{6} 4 \omega^{3} \omega^{23} 2 0 5 \omega^{9}  
&(\ref{eq:recursive3})\\

&7&16&8&-&  [0, 0, 1, 1] &  \omega^{3} &   \omega^{34} \omega^{3} \omega^{2} \omega^{10} \omega^{44} \omega^{35} \omega^{9}  &  \omega^{27} \omega^{33} 4 \omega^{42} \omega^{27} \omega^{19} \omega^{41} 3 \omega^{28 } 
&\\

&8&18&9&-&  [1, 1, 4, 4] &  \omega^{3} &   \omega^{18} \omega^{35} \omega^{7} \omega^{39} 4 \omega^{5} \omega^{12} \omega^{39}  &
   \omega^{35} \omega^{46} \omega^{30} \omega^{6} \omega^{29} \omega^{18} \omega^{11} \omega^{47} \omega^{13} \omega^{23}  
&\\
\\
\hline
\\

&4&10&6&6& [ 0, 1, 0, 1 ]&  1 &   \omega^{52} \omega^{56} \omega^{4} \omega^{20} &   \omega^{26} \omega^{13} \omega^{41} \omega^{8} \omega^{47} \omega^{29}  
&\\

&5&12&7&-&  [0, 1, 1, 1] &  1 &   \omega^{21} \omega^{44} 1 1 1 &  1 \omega^{19} \omega^{37} \omega^{62} \omega^{57} \omega^{37} \omega^{42 }
&\\

8^2&6&14&7&-& [ 0, 1, 0 ,1 ]&  1 &   \omega^{15} \omega^{51} \omega^{40} \omega^{15} \omega^{17} \omega^{55 } &
   \omega^{15} \omega^{39} \omega^{25} \omega^{28} \omega^{56} 1 \omega^{6} \omega^{53}  
&(\ref{eq:recursive3})\\

&7&16&8&-&  [0, 0, 1, 2] &  1 &   \omega^{47} \omega^{45} \omega^{32} \omega^{47} \omega^{44} \omega^{57} \omega^{41} &
  \omega^{50} \omega^{51} \omega^{15} \omega^{40} \omega^{12} \omega^{20} \omega^{23} \omega^{17} \omega^{32}  
&\\

&8&18&9&-& [ 0, 1, 0, 2 ]&  1 &   \omega^{10} \omega^{3} \omega^{9} \omega^{29} \omega^{39} \omega^{36} \omega^{3} \omega^{29}  
&   \omega^{18} \omega^{14} \omega^{59} \omega^{56} \omega^{17} \omega^{15} \omega^{3} \omega^{4} \omega^{18} \omega^{26}  
&\\
\\
\hline
\\

&4&10&6&6&  [0, 0, 1, 1 ]&  \omega^{4 }&   \omega^{70} \omega^{6} \omega^{6} \omega^{69}  &   \omega^{22} \omega^{5} 2 \omega^{48} \omega^{57} \omega^{23}  
&\\

&5&12&7&7&  [0, 0, 0, 1] &  \omega^{4 }&   \omega^{29} \omega^{60} \omega \omega^{68} \omega^{13} 
&   \omega^{10} \omega^{21} \omega^{14} \omega^{32} \omega^{27} \omega^{17} \omega^{58}
&\\

9^2&6&14&7&-& [ 0, 0, 1, 1] &  \omega^{4} &   \omega^{52} \omega^{51} \omega^{14} \omega^{57} \omega^{27} \omega^{39}  
&   \omega^{69} \omega^{45} \omega^{49} \omega \omega^{31} \omega^{62} \omega^{58} \omega^{75}  
&(\ref{eq:recursive3})\\

&7&16&8&-&  [1, 0, 1, 1 ]&  \omega^{4 }&  \omega^{68} \omega^{23} \omega^{15} \omega^{9} \omega^{30} \omega^{14} \omega^{36}  
&  \omega^{5} \omega^{5} \omega^{29} \omega^{65} \omega^{11} \omega^{57} \omega \omega^{25} \omega^{36}  
&\\

&8&18&9&-&  [0, 1, 0, 1 ]&  \omega^{4} &   \omega^{17} \omega^{55} 2 \omega^{21} \omega^{64} \omega^{68} \omega^{47} \omega^{6} 
&  \omega^{52} \omega^{43} \omega^{63} \omega^{69} 0 \omega^{23} \omega^{10} \omega^{32} \omega^{2} \omega^{79} 
&\\
\\
\hline
\\
11^2&6&14&8&-&  [0, 1, 2, 4 ]&  \omega^{5} &   \omega^{25} \omega^{61} \omega^{46} \omega^{86} \omega^{42} \omega^{100}   
&   \omega^{38} \omega^{9} \omega^{93} \omega^{34} \omega^{91} \omega^{91} \omega^{9} \omega^{100}  
&(\ref{eq:recursive3})\\
\\
\hline
\\
13^2&6&14&8&-&  [0, 0, 4, 5] &  \omega^{6 }&   \omega^{82} \omega^{164} 3 \omega^{100} \omega^{138} \omega^{163}   &
   \omega^{166} \omega^{58} \omega^{165} \omega^{114} \omega^{129} \omega^{145} \omega^{24} \omega^{83 } 
&(\ref{eq:recursive3})\\

\end{array}
$$\label{table:F49-81}
\end{table}

\begin{table}
\caption{Hermitian self-dual codes over $\F_{q}, q=11^2,13^2,16^2,17^2,19^2$ compared with \cite{GulKimLee},  where $L=N^iP^jQ^kR^l$ is computed by Algorithm 1 with $m=n$, $(d):$ distance in \cite{GulKimLee},$-$: not available in \cite{GulKimLee}}

$$
\begin{array}{c||cccccc|c}
q&n&\text{length}&\text{Distance}&\text{Distance}\cite{GulKimLee}&[i,j,k,l]&[\delta,\theta,\beta,\alpha,\gamma,a,\lambda]&\text{Construction }\\

\hline\\
&2&4&3&3&&&\\
11^2&3&6&4&4&&&(\ref{eq:sd1})\\
&4&8&5&5&[  0,0,0,1 ]&&\\
\\
\hline
\\
&2&4&3&3&&&\\
13^2&3&6&4&4&&&(\ref{eq:sd1})\\
&4&8&5&5&[  0,0,0,1 ]&&\\
\\
\hline
\\
&2&4&3&3&&&\\
16^2&3&6&4&4&&&(\ref{eq:sd1})\\
&4&8&5&5&[  0,0,1,1 ]&&\\
\\
\hline
\\
&2&4&3&3&&&\\
17^2&3&6&4&4&&&(\ref{eq:sd1})\\
&4&8&5&5&[  0,0,0,2 ]&&\\
\\
\hline
\\
&2&4&3&3&&&\\
19^2&3&6&4&4&&&(\ref{eq:sd1})\\
&4&8&5&5&[  0,0,0,1 ]&&\\
\\
\hline
\\
&4&10&6&6&[  0, 0, 1, 1 ]
&
[  0, \omega^2, \omega^{68}, \omega^{75}, \omega^6, \omega^5, \omega^5 ]&\\
&5&12&7&7&[  1, 2, 1, 0 ]
&
[  0, \omega^4, \omega^{64}, \omega^{73}, \omega^3, \omega^5, \omega^5 ]&\\
11^2&6&14&7&7&[  0, 0, 1, 1 ]
&
[  0, \omega^9, \omega^{57}, \omega^{71}, \omega^7, \omega^5, \omega^5 ]&(\ref{eq:double2})\text{ Theorem } \ref{thm:double1}~ 2)\\
&7&16&8&-&[  0, 0, 2, 2 ]
&
[  0, \omega^7, \omega^{58}, \omega^{70}, \omega^9, \omega^5, \omega^5 ]&\\
&8&18&9&-&[  0, 1, 2, 1 ]
&
[  0, \omega^3, \omega^{51}, \omega^{59}, \omega^4, \omega^5, \omega^5 ]&\\
\\
\hline
\\

&4&10&6&6&[  0, 0, 1, 1 ]
&
[  0, \omega, 1, \omega^7, \omega^7, \omega^6, \omega^6 ]&\\
&5&12&7&7&[  1, 1, 1, 0 ]
&
[  0, \omega, \omega^{165}, \omega^7, \omega^{10}, \omega^6, \omega^3 ]&\\
13^2&6&14&7&7&[  0, 1, 0, 1 ]
&
[  0, \omega, \omega^{167}, \omega^7, \omega^8, \omega^6, \omega^5 ]&(\ref{eq:double2})\text{ Theorem } \ref{thm:double1}~ 3)\\
&7&16&8&-&[  0, 0, 1, 1 ]
&
[  0, \omega, \omega^4, \omega^7, \omega^3, \omega^6, \omega^{10} ]&\\
&8&18&9&-&[  1, 0, 1, 1 ]
&
[  0, \omega, \omega^{164}, \omega^7, \omega^{11}, \omega^6, \omega^2 ]&\\
\\
\hline
\\
&5&10&6&6&[  0,0,1,2 ]&&\\
&6&12&7&7&[  0,0,2,1 ]&&\\
16^2&7&14&8&-&[  5,4,1 ,0]&&(\ref{eq:sd2})\\
&8&16&8&-&[  0,0,1,2 ]&&\\
&9&18&9&-&[  0,1,1,2 ]&&\\
\\
\hline
\\
&4&10&6&6&[  0, 0, 0, 1 ]&[  1, \omega, \omega^{166}, \omega^9, \omega^3, \omega^8, \omega^{174} ]&\\
&5&12&7&7&[  0, 0, 1, 1 ]&[  1, \omega, \omega^{176}, \omega^9, \omega^2, \omega^8, \omega^{184} ]&\\
17^2&6&14&8&-&[  0, 0, 1, 4 ]&[  1, \omega, \omega^{25}, \omega^9, \omega^{10}, \omega^8, \omega^{33} ]&(\ref{eq:double2})\text{ Theorem } \ref{thm:double4}~ 4)\\
&7&16&8&-&[  0, 0, 1, 1 ]&[  1, \omega, \omega^{57}, \omega^9, \omega^{11}, \omega^8, \omega^{65} ]&\\
&8&18&9&-&[  0, 0, 1, 1 ]&[  1, \omega, \omega^{208}, \omega^9, \omega^{15}, \omega^8, 4 ]&\\
\\
\hline
\\
&4&10&6&6&[  0, 0, 0, 2 ]&[  1, \omega, \omega^{294}, \omega^{10}, \omega^{15}, \omega^9, \omega^{303} ]&\\
&5&12&7&7&[  0, 0, 2, 2 ]&[  1, \omega, \omega^{12}, \omega^{10}, \omega^{12}, \omega^9, \omega^{21} ]&\\
19^2&6&14&8&8&[  1, 0, 1, 2 ]&[  1, \omega, \omega^{166}, \omega^{10}, \omega^{17}, \omega^9, \omega^{175} ]&(\ref{eq:double2})\text{ Theorem } \ref{thm:double4}~ 4)\\
&7&16&8&-&[  0, 0, 1, 1 ]&[  1, \omega, \omega^{223}, \omega^{10}, \omega^8, \omega^9, \omega^{232} ]&\\
&8&18&9&-&[  0, 0, 1, 2 ]&[  1, \omega, 10, \omega^{10}, \omega^{3}, \omega^9, \omega^{349} ]&\\
\end{array}
\label{table:new}
$$
\end{table}

\end{document}